\documentclass[11pt]{article}

\usepackage[numbers,sort&compress]{natbib}

\usepackage[letterpaper,margin=1in]{geometry}
\usepackage{xcolor}
\usepackage[colorlinks]{hyperref}
\definecolor{lightblue}{rgb}{0.5,0.5,1.0}
\definecolor{darkred}{rgb}{0.5,0,0}
\definecolor{darkgreen}{rgb}{0,0.5,0}
\definecolor{darkblue}{rgb}{0,0,0.5}
\hypersetup{colorlinks,linkcolor=darkred,filecolor=darkgreen,urlcolor=darkred,citecolor=darkblue}
\usepackage{tabularx}
\usepackage{mathtools}
\usepackage{verbatim}
\usepackage{enumitem}
\usepackage{pifont}
\usepackage{amsthm}
\usepackage{refcount}
\usepackage{amssymb}
\usepackage{amsmath}
\usepackage{complexity}
\usepackage{caption}
\usepackage{subcaption}
\usepackage{nicefrac}
\usepackage{mathtools}
\usepackage[linesnumbered,ruled,commentsnumbered,longend]{algorithm2e}
\usepackage{tikz}
\usetikzlibrary{graphs}
\usetikzlibrary{shapes.misc, positioning}
\usetikzlibrary{shapes,fit}
\usepackage[edges]{forest}
\usetikzlibrary{shapes,arrows,calc,decorations.pathmorphing}
\usepackage{amsmath}
\usepackage[]{todonotes}
\usepackage{pgfplots}

\newcommand{\Traces}{\textsc{Traces}}

\newcommand{\saucy}{\textsc{saucy}}
\newcommand{\dejavu}{\textsc{dejavu}}

\newcommand{\Slift}{\uparrow}
\newcommand{\Sred}{\downarrow}

\SetKwProg{Fn}{function}{}{end}\SetKwFunction{RandomAut}{Automorphisms}
\SetKwProg{Fn}{function}{}{end}\SetKwFunction{RandomIso}{Isomorphism}
\SetKwFunction{FirstLeaf}{FirstLeaf}
\SetKwFunction{RandomWalk}{RandomWalk}
\SetKwFunction{NextNodeDFS}{OneNodeDFS}
\SetKwFunction{NextNodeBFS}{OneNodeBFS}
\SetKwFunction{Some}{Some}
\SetKwFunction{Sift}{Sift}
\SetKwFunction{RandomElement}{RandomElement}
\SetKwFunction{RandomChild}{RandomChild}
\SetKwFunction{NextChild}{NewChild}
\SetKwFunction{NextRandomChild}{NewRandomChild}
\SetKwFunction{Refine}{Refine}
\SetKwFunction{Subtree}{Subtree}
\SetKwFunction{RRef}{Ref}
\SetKwFunction{SSel}{Sel}
\SetKwFunction{Sift}{Sift}
\SetKwFunction{CertifyAutomorphism}{CertifyAutomorphism}
\SetKwFunction{CertifyIsomorphism}{CertifyIsomorphism}

\newcommand{\xmark}{\ding{55}}%

\newcommand{\xparagraph}[1]{\textbf{#1}}



\DeclareMathOperator{\Sym}{Sym}

\DeclareMathOperator{\SymL}{SymL}
\DeclareMathOperator{\SymV}{SymV}

\DeclareMathOperator{\syn}{syn}
\DeclareMathOperator{\sem}{sem}

\DeclareMathOperator{\Aut}{Aut}

\DeclareMathOperator{\Var}{Var}
\DeclareMathOperator{\unit}{unit}

\DeclareMathOperator{\Lit}{Lit}

\newtheorem{lemma}{Lemma}
\newtheorem{corollary}[lemma]{Corollary}

\newtheorem{definition}[lemma]{Definition}

\newcommand{\n}[1]{\overline{#1}}

\title{SAT Preprocessors and Symmetry} 
\author{Markus Anders}

\newcommand\blfootnote[1]{%
  \begingroup
  \renewcommand\thefootnote{}\footnote{#1}%
  \addtocounter{footnote}{-1}%
  \endgroup
}
\begin{document}
\maketitle

\begin{abstract}
Exploitation of symmetries is an indispensable approach to solve certain classes of difficult SAT instances. 
Numerous techniques for the use of symmetry in SAT have evolved over the past few decades.
But no matter how symmetries are used precisely, they have to be detected first. 
We investigate how to detect more symmetry, faster.
The initial idea is to reap the benefits of SAT preprocessing for symmetry detection.
As it turns out, applying an off-the-shelf preprocessor before handling symmetry runs into problems: the preprocessor can haphazardly remove symmetry from formulas, severely impeding symmetry exploitation. 

Our main contribution is a theoretical framework that captures the relationship of SAT preprocessing techniques and symmetry.
Based on this, we create a symmetry-aware preprocessor that can be applied safely before handling symmetry.
We then demonstrate that applying the preprocessor does not only substantially decrease symmetry detection and breaking times, but also uncovers hidden symmetry not detectable in the original instances.
Overall, we depart the conventional view of treating symmetry detection as a black-box, 
presenting a new application-specific approach to symmetry detection in SAT.
\blfootnote{The research leading to these results has received funding from the European Research Council (ERC) under the European Union's Horizon 2020 research and innovation programme (EngageS: grant agreement No.~{820148}).}
\end{abstract}

\section{Introduction}
Many difficult classes of SAT instances contain a large number of symmetries.
Exploitation of symmetries is an indispensable tool to speed up solving these instances. 
Various techniques, such as for example the use of symmetry breaking predicates or symmetry-based DPLL branching rules, have evolved over the past few decades \cite{DBLP:series/faia/Sakallah09, DBLP:conf/sat/KatebiSM10}. 
In fact, there is still ongoing research on how to use symmetries best while solving SAT \cite{DBLP:conf/sat/Devriendt0BD16, DBLP:conf/tacas/MetinBCK18, DBLP:conf/eecs/Treethanyaphong18, DBLP:journals/ijait/TchindaD19}.

In practice, state-of-the-art symmetry exploitation is based on syntactic symmetries of the formula \cite{DBLP:series/faia/Sakallah09, DBLP:conf/sat/Devriendt0BD16, DBLP:conf/tacas/MetinBCK18, DBLP:conf/eecs/Treethanyaphong18, DBLP:journals/ijait/TchindaD19}.
Syntactic symmetries are permutations of variables (or literals) that map a formula $F$ back to itself, i.e., a permutation $\varphi$ is a syntactic symmetry whenever $\varphi(F) = F$ holds.
The most common way to compute these symmetries is to first model the given input formula as a graph.
Then, the automorphism group of this model graph is computed using a graph isomorphism solver (e.g., \cite{DBLP:conf/dac/DargaLSM04, DBLP:conf/esa/AndersS21,McKay201494}). 
In fact, computing syntactic symmetries is polynomial-time equivalent to the graph isomorphism problem \cite{Crawford1992ATA}.

Generally speaking, this describes the two major components that make up the use of symmetry in SAT: 
first, there is symmetry detection, i.e., finding symmetries of the formula in the first place. 
Secondly, there is the exploitation of symmetries itself, i.e., using symmetries to cut away parts of the search space. 
In this paper, we investigate and improve the former: how to detect more symmetry, faster. 

\begin{table}[t]
	\centering
\begin{tabular}{|l|c|c|c|c|c|}\hline
& \textsc{C} & \multicolumn{2}{c|}{\textsc{B+C}} & \multicolumn{2}{c|}{\textsc{P+B+C}}\\
instance   &  $T_{solve}$ & \multicolumn{1}{c}{$T_{solve}$} & \#syms & \multicolumn{1}{c}{$T_{solve}$} & \#syms \\\hline
php(9,8)   &  6.74s & $<$1s & $1.46 \cdot 10^{10}$ & $<$1s  & $2.07 \cdot 10^6$ \\
php(15,14) &  $>$60s  & $<$1s & $1.14 \cdot 10^{23}$ & $>$60s & $8.72 \cdot 10^{10}$\\
php(21,20) &  $>$60s  & $<$1s & $1.24 \cdot 10^{38}$ & $>$60s & $2.43 \cdot 10^{18}$\\\hline
\end{tabular}
\caption{Pigeonhole principle solved with \textsc{cryptominisat} (\textsc{C}), \textsc{BreakID}+\textsc{cryptominisat} (\textsc{B+C}), and \textsc{cryptominisat} (preprocessor)+\textsc{BreakID}+\textsc{cryptominisat} (\textsc{P+B+C}).} \label{tab:php}
\end{table}
Due to the sheer size of SAT instances and number of symmetries, symmetry detection can indeed become expensive. 
So much so, that for the state-of-the-art symmetry breaking tool \textsc{BreakID}, the version that limits the time used for symmetry detection outperforms the version that is not time-limited \cite{DBLP:conf/sat/Devriendt0BD16}.
Often, the reason why handling symmetry is slow, is that the underlying instances are bloated with many easily reducible variables and clauses.
Some symmetry detection tools try to rectify this by using intricate graph-level preprocessing techniques.
In fact, these techniques have been meticulously engineered to deal with structures common in CNF formulas, i.e., low-degree vertices 
\cite{DBLP:conf/dac/DargaLSM04, DBLP:conf/dac/DargaSM08, McKay201494, DBLP:journals/corr/abs-0804-4881}.
While this marks the currently most successful approach to symmetry detection for CNF formulas, it does not yet exploit the underlying semantics of SAT.

In the SAT domain, the typical first step in tackling large formulas would be to first apply a SAT preprocessor to reduce the formulas \cite{DBLP:conf/sat/EenB05}.
This naturally leads to the question: could applying a SAT preprocessor \emph{before} symmetry detection reduce computation time, and maybe even aid in finding more symmetry?
Indeed, if applying a SAT preprocessor were to improve symmetry detection in any way, then this could be considered a win-win situation: SAT preprocessing techniques are usually applied anyway and hence the resulting improvement would be virtually free.
Unfortunately, after conducting simple testing (see Table~\ref{tab:php}) we find that applying an off-the-shelf preprocessor 
(e.g., \cite{DBLP:conf/sat/SoosNC09}) 
\emph{before} dealing with symmetry (e.g., \cite{DBLP:conf/sat/Devriendt0BD16}) does not work for symmetry detection.
The test suggests that even in the basic case of the pigeonhole principle, the preprocessor \emph{removes} symmetries in a way that renders the subsequent symmetry breaking ineffective.

On the other hand, preprocessing can actually also lead to \emph{more} symmetry:
consider the CNF formula $(x) \wedge (\n{x} \vee a \vee c) \wedge (b \vee c)$.
Without any alteration, the formula has no non-trivial symmetries.
However, when we apply, say, the unit rule on $x$, the formula becomes $(a \vee c) \wedge (b \vee c)$.
This in turn makes $a$ and $b$ symmetrical.
Indeed, in this example, simplifying the formula allows us to detect more symmetry.

Confusingly, we thus find that SAT preprocessing can both lead to less, as well as more symmetry in formulas.
This raises several questions.
What is the right order of operations? 
Can we choose and schedule SAT preprocessing techniques in such a manner, that we both increase interesting symmetry while making the formula easier for symmetry detection?
Overall, it seems that a fundamental understanding of the effect of preprocessing techniques on symmetry is required.

The ultimate goal of this line of research is to develop symmetry detection for SAT that maximizes detected symmetries while minimizing computation time.

	\xparagraph{Contribution.}
	We improve symmetry detection of CNF formulas through the use of \emph{adapted} SAT preprocessing: we reduce the time needed for state-of-the-art symmetry detection and breaking, while also uncovering hidden symmetry not detectable in the unprocessed formulas.

	In particular, our main contribution is to provide a theoretical framework that captures how symmetries of formulas before and after applying CNF transformations are related.  
	Among other results, the most important property we consider is whether transformations are ``symmetry-preserving'': 
	here, we demand that all applicable symmetries of the original formula are also symmetries of the reduced formula (see Section~\ref{sec:transformations_and_symmetry} for a formal definition).
	This means, when applying a symmetry-preserving transformation, we only have to compute symmetries of the reduced formula.
	This in turn allows us to categorize a selection of quintessential SAT preprocessing techniques  (see first column of Table~\ref{fig:results}).
	For techniques that turn out to not be symmetry-preserving, we provide tailored restrictions to rectify this.

	The novelty of our approach lies in the fact that we exploit \emph{SAT techniques} to improve \emph{symmetry detection}: we depart the conventional view of treating symmetry detection as a black-box, 
	presenting a new application-specific approach to symmetry detection in SAT. 

	\xparagraph{Theoretical Framework.}
	On the theoretical side, the first challenge is to define formal notions describing the effect of CNF transformations on symmetry.
	We identify three main properties that seem particularly interesting. 
	These properties describe how symmetries behave across the different directions of the transformations.
	In the following, we want to give a brief intuition for these properties.
	For the formal definitions see Section~\ref{sec:transformations_and_symmetry}. 
	Let $F$ denote a CNF formula that is then transformed into $F'$ (using the transformation in question, e.g., applying the unit rule). We demand that $\Var(F') \subseteq \Var(F)$, where $\Var$ denotes the set of variables of a formula.
	\begin{enumerate}
	\item 
	If a transformation is \emph{symmetry-preserving (SP)}, then any syntactic 
	symmetry of $F$ is also a
	syntactic 
	symmetry of $F'$, when restricted to the reduced set of variables $\Var(F')$.
	When applying the transformation, it therefore suffices to compute symmetries of $F'$. 

	Regarding the particular example in Table~\ref{tab:php}, we want to remark that this property indeed guarantees that a transformation must either preserve \emph{all} the symmetries of the pigeonhole principle instances, or reduce \emph{all} variables at once, in which case deciding SAT becomes trivial.

	\item 
	If a transformation is \emph{weakly symmetry-preserving (WSP)}, then it is possible to restrict syntactic 
	symmetries of $F$ using group-theoretic algorithms to \emph{semantic}
	symmetries of transformed formulas of $F'$.
	This allows for a manual collection of symmetries on reduced formulas, while not guaranteeing an automatic preservation of all the symmetries.  

	\item 
	If a transformation is \emph{symmetry-lifting (SL)}, then semantic symmetries of $F'$ are also semantic symmetries of the original formula $F$.
	This means that the transformation can be used to potentially find more symmetries of the original formula.
	This seems particularly useful for model counting and enumeration tasks, i.e., applications in which all solutions of the original formula are of interest. In other words, this enables us to potentially find more symmetries of $F$ through the transformation, without actually having to apply the transformation to $F$.
	\end{enumerate}
	The results of our theoretical analysis are summarized in Table~\ref{fig:results}.
	For techniques that turn out to not be symmetry-preserving, we give new, restricted variants that are. 

	\xparagraph{Practical Evaluation.}
	On the practical side, armed with these new insights, we analyze the effect of symmetry-preserving SAT preprocessing techniques on symmetry detection.
	We evaluate symmetry detection on instances from the main track of the SAT competition 2021 \cite{satComp2021}.
	We simplify the formulas using only the aforementioned symmetry-preserving transformations.
	It then turns out that even when only applying this subset of simplifications, computation times of state-of-the-art symmetry detection and breaking algorithms are decreased substantially.
	In terms of the type of symmetry detected, we observe that indeed, preprocessing is frequently able to uncover symmetry that is hidden in the unprocessed instance.
	Moreover, it turns out that unprocessed instances often contain a substantial number of symmetries that exclusively operate on variables that can be reduced away through preprocessing techniques.
	\begin{table}[t]
		\centering
		\begin{tabular}{|l|r|r|r|}
			\hline
			transformation & \footnotesize SP &  WSP &\footnotesize SL \\\hline
			subsumption$^1$ & $\checkmark$ & \checkmark & $\checkmark$ \\
			self-subsumption & \xmark & \checkmark & $\checkmark$\\
			simultaneous self-subsumption$^2$ & \checkmark & \checkmark & $\checkmark$\\
			adding learned clauses & \xmark & \checkmark & \checkmark  \\
			unit$^1$ & $\checkmark$ & \checkmark & $\checkmark$ \\
			pure$^1$ & $\checkmark$  & \checkmark &  \xmark \\
			blocked clause elimination$^1$ & $\checkmark$  & \checkmark &  \xmark \\
			bounded variable elimination  & \xmark  & \checkmark &  \xmark \\
			symmetric variable elimination$^2$  & \checkmark  & \checkmark &  \xmark \\
			\hline
		\end{tabular}
		\caption{The table shows whether symmetries are preserved under applying a given CNF transformation (SP and WSP) and whether symmetries of a transformed formula lift back to the original one (SL). $^1$ indicates that it is assumed that the transformation is applied exhaustively. $^2$ denotes new, symmetry-preserving variants defined in Section~\ref{sec:transformations_and_symmetry}.
		} \label{fig:results}
	\end{table}%

\section{Syntactic and Semantic Symmetry} \label{sec:syntactic_and_semantic_symmetry}
We begin by introducing some general notation used in SAT solving. 
This allows us to discuss the notions of syntactic and semantic symmetries, which we will make use of in this paper. 

\xparagraph{SAT.} 
Given a variable $v$, we define the corresponding \emph{literals} $v$ and $\n{v}$, the latter denoting negation.
A SAT instance $F$ is commonly given in \emph{conjunctive normal form} (CNF), i.e., a conjunction of disjunctions of literals, meaning that it is in the form
$F = \bigwedge_{i \in \{1, \dots{}, m\}} \bigvee_{j \in \{1, \dots{}, k_i\}} l_{i,j}$
where 
	each disjunction of literals is called a \emph{clause},
	$k_i$ is the number of literals in clause $i$
	and $l_{i, j}$ is the $j$-th literal in the $i$-th clause.
\noindent In this paper, we use the common representation where a formula in CNF is a set of clauses and clauses are sets of literals, without explicitly encoding the conjunctions and disjunctions. 
The instance above thus becomes 
$F = \{\{l_{1,1},\dots{},l_{1,k_1}\}, \dots{}, \{l_{m,1},\dots{},l_{m,k_m}\}\}$.

We denote with $\Var(F) = \{v_1, \dots{}, v_n\}$ the set of variables of $F$. 
We implicitly use the fact that $\n{\n{v}} = v$ whenever possible:
given a literal $l = \n{v}$ we may write $\n{l} = \n{\n{v}} = v$. 
We never distinguish between $\n{\n{v}}$ and $v$.
The set of literals of a formula $F$ is denoted by $\Lit(F) := \Var(F) \cup \{\n{v} \; | \; v \in \Var(F)\}$.

An \emph{assignment} of $F$ is a function $\sigma : L \to \{\bot, \top\}$ where $L \subseteq \Lit(F)$. An assignment must always act consistently on the literals of a variable, meaning that $\sigma(v) = \top$ if and only if $\sigma(\n{v}) = \bot$ holds.
We call an assignment \emph{complete} whenever $L = \Lit(F)$ and \emph{partial} otherwise. 
Let $v \in \Var(F)$. 
Abusing notation, we may write that $l \in \sigma$ whenever $\sigma(l) = \top$, and $l \notin \sigma$ otherwise. 
Note that if $l \notin \sigma$ and $\n{l} \notin \sigma$, then $l \notin L$.

We can simplify a formula $F$ with respect to $\sigma$:
$F[\sigma] := \{C[\sigma] \;|\; C \in F \wedge \nexists l \in \Lit(F): (l \in \sigma \wedge l \in C)\}$ with
$C[\sigma] := \{l \;|\; l \in C \wedge \n{l} \notin \sigma\}$.
If $F[\sigma] = \{\}$ we call $\sigma$ a \emph{satisfying assignment}, whereas if $\{\}\in F[\sigma]$ we call $\sigma$ a \emph{conflicting assignment}. 
A complete assignment is either satisfying or conflicting.

Let $C_1$ and $C_2$ denote two clauses with $x \in C_1$ and $\n{x} \in C_2$. 
We denote with $C_1 \circ_x C_2 = (C_1 \smallsetminus \{x\}) \cup (C_2 \smallsetminus \{\n{x}\})$ the resolvent of $C_1$ and $C_2$ on variable $x$. 

\xparagraph{Syntactic Symmetry.} 
We introduce the notion of \emph{syntactic symmetries} of a CNF formula $F$. 
Consider bijections $\varphi: \Lit(F) \to \Lit(F)$ mapping literals to literals.
A bijection $\varphi$ naturally lifts to clauses, formulas and assignments, by applying $\varphi$ element-wise to these objects. 
Syntactic symmetries have two defining properties:
(1) The formula is mapped back to itself, meaning $\varphi(F) = F$ holds.
(2) For all $l \in \Lit(F)$ it holds that $\n{\varphi(l)} = \varphi(\n{l})$, meaning $\varphi$ also induces a permutation of the variables. 
If the above requirements are met, we call $\varphi$ a syntactic symmetry or automorphism of $F$. These symmetries form a permutation group under composition. We denote this group of automorphisms of $F$ by $\Aut_{\syn}(F)$. 

By modelling a CNF formula as a graph (and vice versa), it can be shown that computing all syntactic symmetries is polynomial-time equivalent to computing the automorphism group of a graph and hence to the graph isomorphism problem \cite{Crawford1992ATA}, also known as
\GI-complete.

Therefore, syntactic symmetries can indeed be computed using tools for graph isomorphism (for modern tools see, e.g., \cite{McKay201494,DBLP:conf/esa/AndersS21, Darga:2004:ESS:996566.996712}). 

\xparagraph{Semantic Symmetry.} Let us again consider bijections $\varphi: \Lit(F) \to \Lit(F)$.
We define the notion of a \emph{semantic symmetry} $\varphi$ as follows:  
(1) For all complete variable assignments $\sigma$ of $F$ it holds that $F[\sigma] = \varphi(F)[\sigma]$. 
(2) For all $l \in \Lit(F)$ it holds that $\n{\varphi(l)} = \varphi(\n{l})$. 
Note that the second property is the same property as for syntactic symmetries.

We may apply a symmetry $\varphi$ on a complete assignment $\sigma$: we let $\varphi(\sigma)$ denote the complete assignment with
$\varphi(\sigma)(v) := \sigma(\varphi(v))$.
We can interchange applying any semantic symmetry $\varphi$ to a formula $F$ or a corresponding complete assignment $\sigma$, i.e., $\varphi(F)[\sigma] = F[\varphi(\sigma)]$.

The set of all semantic symmetries indeed also forms a group under composition as well: 
\begin{lemma} The set of all semantic symmetries forms a group under composition. \label{lem:semsym_is_group}
\end{lemma}
\begin{proof}
	Consider two semantic symmetries $\varphi, \varphi'$ of a formula $F$. 
	We show that $\varphi\circ\varphi'$ is a semantic symmetry of $F$ as well. 
	By definition, for any complete variable assignment $\sigma$ of $F$, it holds that $F[\sigma] = \varphi(F)[\sigma] = \varphi'(F)[\sigma]$.  
	But since the definition of semantic symmetry requires that this holds for any $\sigma$, it in particular also holds for any $\varphi(\sigma)$. 
	Hence, $F[\sigma] = \varphi'(F)[\sigma] = F[\varphi'(\sigma)] = \varphi(F)[\varphi'(\sigma)] = \varphi'(\varphi(F))[\sigma] = \varphi\circ\varphi'(F)[\sigma]$.
	\end{proof}
We denote the permutation group of all semantic symmetries of $F$ as $\Aut_{\sem}(F)$. 

In order to streamline notation, we introduce the notion of a ``symmetric group'' on a set of literals $L$. 
We want to be able to denote the symmetric group, but with the restriction $\n{\varphi(l)} = \varphi(\n{l})$, as is required for syntactic and semantic symmetries. 
We therefore denote $\SymL(L) := \{\varphi \in \Sym(L) \;|\; \forall l \in L: \n{\varphi(l)} = \varphi(\n{l})\}$ (where $\Sym(L)$ denotes the symmetric group on the domain $L$).
We also give an alternative version for a set of literals $V$ where there is no $l \in V$ with $\n{l} \in V$, which disallows negation symmetries: 
$\SymV(V) := \{\varphi' \; | \; \varphi \in \Sym(V)\}$, where 
$\varphi'(l) := \varphi(l)$ if $l \in V$, and $\varphi'(l) := \n{\varphi(\n{l})}$ otherwise.

It is easy to see that all syntactic symmetries are also semantic symmetries, i.e., $\Aut_{\syn}(F) \subseteq \Aut_{\sem}(F)$, but the reverse does not hold in general \cite{DBLP:series/faia/Sakallah09}. 
As an example note that for any unsatisfiable formula $F$, the semantic symmetry group is actually $\SymL(\Lit(F))$.
Unsurprisingly, computing semantic symmetries is \NP-hard.

\xparagraph{Domain Change.}  Next, we define tools to compare groups on different sets of literals.
This is needed whenever we are, for example, comparing the symmetries of a reduced version of a formula to the original formula. 
To enable this, we define a way to reduce or lift the (finite) domain of a given group. 
Let $\Gamma$ be a permutation group on the domain $\Omega$, and let $\Omega'$ be a set such that $\Omega \subseteq \Omega'$.
We then define $\Gamma\Slift^{\Omega'} := \{\varphi\Slift^{\Omega'} \;|\; \varphi \in \Gamma\}$ where $\varphi\Slift^{\Omega'}: \Omega' \to \Omega'$ with
$\varphi\Slift^{\Omega'}(x) := \varphi(x) \text{ if } x \in \Omega$ and $\varphi\Slift^{\Omega'}(x) := x$ else.
Intuitively, we just extend every permutation on $\Omega$ with the identity on $\Omega' \smallsetminus \Omega$.

Analogously, to reduce the domain, let $\Omega' \subseteq \Omega$. 
We first define the \emph{setwise stabilizer} $\Gamma_{\{\Omega'\}} := \{\varphi \;|\; \varphi \in \Gamma \wedge \varphi(\Omega') = \Omega'\}$.
We further define $\Gamma\Sred_{\Omega'} := \{\varphi|_{\Omega'} \;|\; \varphi \in \Gamma_{\{\Omega'\}}\}$, i.e., first taking the setwise stabilizer fixing $\Omega'$ and subsequently reducing the domain to $\Omega'$.
Lastly, we also recall the \emph{pointwise stabilizer} $\Gamma_{(\Omega')} := \{\varphi \in \Gamma \;|\; \forall p \in \Omega': \varphi(p) = p\}$, fixing all points of $\Omega'$ individually.

\section{Transformations and Symmetry} \label{sec:transformations_and_symmetry}
We now focus on the main objective of this paper:
we analyze the relationship of CNF transformations and symmetry.
In order to do this systematically, we first describe a prototype for all the transformations that we consider.
A \emph{transformation} $\Pi$ defines for a formula $F$ a set of formulas $\Pi(F) = \{F_1, \dots{}, F_m\}$, where for each formula $F_i$ it holds that $\Var(F_i) \subseteq \Var(F)$ ($i \in \{1, \dots{}, m\}$).
We denote $F \xrightarrow{\Pi} F'$ whenever $F' \in \Pi(F)$, i.e., we \emph{apply} $\Pi$ to $F$.
Naturally, a transformation can be non-deterministic: at any given point, potentially, we could transform $F$ into any of the formulas of $\Pi(F)$.
We write $F \xrightarrow{\Pi}^* F^*$ whenever $\Pi$ is applied exhaustively, i.e., until $\Pi(F^*) = \emptyset$.
We always assume that both the set $\Pi(F)$ as well as the sequence of rule applications is finite.

Next, we define the properties of interest. 
Firstly, we define when transformations are \emph{symmetry-preserving}. This property formalizes the notion that all symmetries of the original formula that operate on remaining literals are also syntactic symmetries of the reduced formulas.
\begin{definition}[Symmetry-preserving.] Let $\Pi$ denote a 
	transformation.
	We call $\Pi$ \emph{symmetry-preserving}, if $\Aut_{\syn}(F)_{\{\Lit(F')\}} = \Aut_{\syn}(F)$ and $\Aut_{\syn}(F)\Sred_{\Lit(F')} \leq \Aut_{\syn}(F')$ hold for all CNF formulas $F, F'$ with $F \xrightarrow{\Pi} F'$.
\end{definition}
The symmetry-preserving property indeed comprehensively guarantees preservation of syntactic symmetry on variables of $F'$.
When we apply the simplification, we only need to compute symmetries of $F'$.
Despite the name, on the removed variables $\Var(F) \smallsetminus \Var(F')$, naturally, $F'$ does not contain the symmetries of $F$: it can be the case that 
$|\Aut_{\syn}(F)\Sred_{\Lit(F')}| < |\Aut_{\syn}(F)|$.
We refer to these symmetries of $F$ as \emph{reducible symmetries}.

Reducible symmetries \emph{solely} permute literals removed by the transformation.
All symmetries that operate on remaining literals present in $F$ -- even the ones that simultaneously permute remaining literals and removed literals -- are guaranteed to be preserved in $F'$.
In fact, there even are no symmetries permuting remaining literals with removed literals to begin with (since $\Aut_{\syn}(F)_{\{\Lit(F')\}} = \Aut_{\syn}(F)$).

We want to mention that whether applying the transformation is ``desirable'' in the first place is clearly not captured by the symmetry-preserving property.
For example, applying blocked clause elimination (which is symmetry-preserving, see Section~\ref{sec:solution_alter}) can indeed make formulas more difficult to solve \cite{DBLP:journals/dam/Kullmann99}.
This in turn obscures the question whether blocked clause elimination should be applied before handling symmetry, since it is not clear whether it should be applied at all.
This discussion is however unrelated to our goal: in any case, the property guarantees that if we apply the transformation, we only need to compute symmetries on the remaining formula.

To the contrary, $ \Aut_{\syn}(F')$ can also contain additional symmetries, which do not appear in $\Aut_{\syn}(F)\Sred_{\Lit(F')}$.
We refer to these symmetries as \emph{hidden symmetries}.

Next, we define a weaker notion of the symmetry-preserving property:
\begin{definition}[Weakly Symmetry-preserving.] Let $\Pi$ denote a transformation. We call $\Pi$ \emph{weakly symmetry-preserving}, if $\Aut_{\syn}(F)\Sred_{\Lit(F')} \leq \Aut_{\sem}(F')$ holds for all CNF formulas $F, F'$ with $F \xrightarrow{\Pi} F'$.
\end{definition}
Here, symmetries do not have to be preserved syntactically, only semantically. Also, the setwise stabilizer may remove symmetries beyond reducible symmetries. 

Lastly, we also want to reconcile symmetries in the opposite direction of transformations, i.e., how symmetries of reduced formulas relate to symmetries of the original formula.
We define the \emph{symmetry-lifting} property, which formalizes the notion that symmetries of the reduced formula are symmetries of the original formula: 
\begin{definition}[Symmetry-lifting.] Let $\Pi$ denote a transformation. We call $\Pi$ \emph{symmetry-lifting}, if $\Aut_{\sem}(F')\Slift^{\Lit(F)} \leq \Aut_{\sem}(F)$ holds for all CNF formulas $F, F'$ with $F \xrightarrow{\Pi} F'$. 
\end{definition}
If the transformed formula contains syntactic symmetries that the original one does not, we therefore indeed uncover semantic symmetries of the original formula.
In particular, concerning the discussion above, this also makes the additional symmetries found independent of whether we actually want to apply the transformation or not: we can use symmetries in the original formula, without having to apply the transformation 

We now show a technical lemma that will aid us in proving that transformations are symmetry-preserving. But in order to state the lemma, we first need to recall and adapt two well-known properties: isomorphism-invariance and confluence.

We say a transformation $\Pi$ is \emph{isomorphism-invariant}, whenever for all formulas $F$, for all finite sets $V$, for all $\varphi \in \SymL(\Lit(F) \cup V \cup \{\n{v} \;|\; v \in V\})$ and for all $F' \in \Pi(F)$ it holds that $\varphi(F') \in \Pi(\varphi(F))$.
It follows that if $\Pi$ is isomorphism-invariant and $\varphi \in \Aut_{\syn}(F)$, then $F' \in \Pi(F)$ implies $\varphi(F') \in \Pi(F)$.
Intuitively, the property states that transformations are invariant under renaming of variables. The set $V$ represents variable names not present in $F$.

Furthermore, a transformation $\Pi$ is \emph{confluent} whenever the following holds for any $F$: let $F^*_1$ and $F^*_2$ denote formulas where $\Pi$ was applied exhaustively to $F$, i.e., $F \xrightarrow{\Pi}^* F^*_1$ and $F \xrightarrow{\Pi}^* F^*_2$. It then must follow that $F^*_1 = F^*_2$.
Using these two properties, we can prove our lemma:
\begin{lemma} \label{lem:con_iso_implies_sympreserve}  Let $\Pi$ be a transformation. If $\Pi$ is confluent and isomorphism-invariant, it follows that the transformation that applies $\Pi$ exhaustively is symmetry-preserving.
\end{lemma}
\begin{proof}
	Let $F^*$ denote the formula where $\Pi$ was applied exhaustively to $F$, i.e., $F \xrightarrow{\Pi}^* F^*$.
	Since $\Pi$ is confluent, $F^*$ is indeed unique.
	We need to show that the transformation $\Pi'$ that defines $F \xrightarrow{\Pi'} F^*$ is symmetry-preserving, i.e., that both $\Aut_{\syn}(F)\Sred_{\Lit(F^*)} \leq \Aut_{\syn}(F^*)$ and $\Aut_{\syn}(F)_{\{\Lit(F^*)\}} = \Aut_{\syn}(F)$ hold.

	Let $F \xrightarrow{\Pi} F'$ and $\varphi \in \Aut_{\syn}(F)$.
	Since $\varphi$ is a symmetry of $F$, it holds that $\varphi(F) = F$. 
	Let $F \xrightarrow{\Pi} F_1 \xrightarrow{\Pi} \dots{} F^*$ denote a derivation of $F^*$. It follows due to isomorphism-invariance that we may also derive $F = \varphi(F) \xrightarrow{\Pi} \varphi(F_1) \xrightarrow{\Pi} \dots{} \varphi(F^*)$. Due to confluence, we know that $F^* = \varphi(F^*)$.
	But this proves precisely that $\varphi|_{\Lit(F^*)}$ is a syntactic symmetry of $F^*$, i.e., $\varphi|_{\Lit(F^*)} \in \Aut_{\syn}(F)_{\{\Lit(F^*)\}}$.
	Since this is true for all $\varphi \in \Aut_{\syn}(F)$, it also follows that $\Aut_{\syn}(F) = \Aut_{\syn}(F)_{\{\Lit(F^*)\}}$, proving the claim.
\end{proof}

\subsection{Equivalence-Preserving Transformations} \label{sec:solution_preserve}
Let us first consider transformations that do not manipulate the set of solutions, i.e., with the property that for all $F \xrightarrow[]{\Pi} F'$ and complete assignments $\sigma$ of $F$, $F[\sigma] = F'[\sigma]$ holds. 
We call such transformations \emph{equivalence-preserving}. 
We show the following useful property.
\begin{lemma} \label{lem:variable_preserve_semantic}  Let $\Pi$ denote an equivalence-preserving transformation. Let $F, F'$ be CNF formulas with $F \xrightarrow{\Pi} F'$. 
Then, it holds both that $\Aut_{\sem}(F)\Sred_{\Lit(F')} \subseteq \Aut_{\sem}(F')$ and $\Aut_{\sem}(F')\Slift^{\Lit(F)} \subseteq \Aut_{\sem}(F)$.
\end{lemma} 
\begin{proof} Let $F \xrightarrow[]{\Pi} F'$. 
	The equivalence-preserving property of $\Pi$ guarantees that for all assignments $\sigma$ of $F$, it holds that $F'[\sigma] = F[\sigma]$.  
	Consider any semantic symmetry $\varphi$ of $F'$ extended by the identity, i.e., $\varphi \in \Aut_{\sem}(F')\Slift^{\Lit(F)}$.
	It follows that both $F'[\sigma] = F[\sigma]$ as well as $F'[\varphi(\sigma)] = F[\varphi(\sigma)]$ for all complete assignments $\sigma$ of $F$.
	We can conclude $F'[\varphi(\sigma)] = F'[\sigma] = F[\sigma] = F[\varphi(\sigma)]$, meaning $\varphi \in \Aut_{\sem}(F)$.
	On the other hand, using the same argument, it also follows that any semantic symmetry $\varphi \in \Aut_{\sem}(F)\Sred_{\Lit(F')}$ is a semantic symmetry of $F'$.
\end{proof}
Considering our terminology, this proves that equivalence-preserving transformations are always symmetry-lifting and weakly symmetry-preserving.

\xparagraph{Subsumption.} Let us now consider the \emph{subsumption} rule as our first concrete instance of a transformation rule.
We formally define the subsumption rule:
\[\frac{C \in F \quad D \in F \quad C \subset D}{F \smallsetminus \{D\}}.\]
It is easy to see that subsumption is equivalence-preserving, thus, Lemma~\ref{lem:variable_preserve_semantic} applies.
The question remains whether subsumption is symmetry-preserving.
While it is not difficult to show that subsumption is not symmetry-preserving in general, it is so when applied exhaustively and can be concluded using Lemma~\ref{lem:con_iso_implies_sympreserve}:
\begin{lemma} \label{lem:subsumption_exhaustive} Exhaustive subsumption is symmetry-preserving.
\end{lemma}
\begin{proof} 
	We prove the claim by applying Lemma~\ref{lem:con_iso_implies_sympreserve}. 
	First of all, note that subsumption is indeed confluent \cite{DBLP:journals/jair/HeuleJLSB15}.
	It thus suffices to show that subsumption is isomorphism-invariant, which follows readily: whenever a clause $C$ is subsumed by $D \in F$, then indeed, $\varphi(C)$ is subsumed by $\varphi(D) \in \varphi(F)$ for all $\varphi \in \SymL(\Lit(F))$.
\end{proof}
Throughout the paper, we often show that applying a rule exhaustively is symmetry-preserving. We want to remark however that this is usually not a requirement and there are indeed also other ways to apply rules in a symmetry-preserving manner (e.g., in a round-based manner).

\xparagraph{Self-Subsumption.} Next, we analyze self-subsumption \cite{DBLP:conf/sat/EenB05}. 
We define self-subsumption based on \emph{self-subsuming resolution} \cite{DBLP:conf/sat/EenB05}: 
\[\frac{C_1 \cup \{x\} \quad C_2 \cup \{\n{x}\} \quad C_1 \subset C_2}{C_2}\]
Whenever we can apply self-subsuming resolution to a formula $F$, we can transform $F$ by removing $C_2 \cup \{\n{x}\}$ and adding $C_2$ (essentially removing the literal $\n{x}$ from $C_2$).
Since this simply constitutes one application of resolution and subsumption, self-subsumption is naturally equivalence-preserving.  
Thus, Lemma~\ref{lem:variable_preserve_semantic} is applicable.

Self-subsumption (also exhaustively, or in conjunction with subsumption) is neither confluent nor symmetry-preserving. Consider the formula
\[(a_1 \vee b_1 \vee \n{x_1} \vee \n{y_1}) \wedge (a_1 \vee x_1) \wedge (b_1 \vee \n{x_1} \vee y_1) \wedge (a_2 \vee b_2 \vee \n{x_2} \vee \n{y_2}) \wedge (a_2 \vee x_2) \wedge (b_2 \vee \n{x_2} \vee y_2).\] 
Clearly, $a_1$ and $a_2$ are symmetrical. Now, one way to exhaustively apply self-subsumption is
\[(a_1 \vee b_1 \vee \n{y_1}) \wedge (a_1 \vee x_1) \wedge (b_1 \vee \n{x_1} \vee y_1) \wedge (a_2 \vee b_2 \vee \n{x_2}) \wedge (a_2 \vee x_2) \wedge (b_2 \vee \n{x_2} \vee y_2).\]  
In this reduced formula however, $a_1$ and $a_2$ are not symmetrical anymore.

\xparagraph{Simultaneous Self-Subsumption.} We present a way to make self-subsumption symmetry-preserving. 

We change two aspects of self-subsumption.
Firstly, we only apply self-subsumption whenever there is only one unique literal that can be removed from a clause.
Secondly, we apply all applicable self-subsumptions simultaneously, in a round-based scheme.

Let $f_F(C_1, C_2, x)$ denote whenever self-subsuming resolution with respect to literal $x$ is applicable to $C_1, C_2$ in $F$ (i.e., removing the literal $x$ from $C_2$).
We restrict self-subsuming resolution to \emph{unique self-subsuming resolution}:
\[\frac{f_F(C_1, C_2, x) \quad \forall x' \in \Lit(F), \forall C_1' \in F:  f_F(C_1', C_2, x') \implies x = x'}{C_2 \smallsetminus \{x\}}\]
If unique self-subsuming resolution is applicable, then there is only one unique literal that can be removed from $C_2$ in $F$ using self-subsumption.
Hence, we can define a deterministic function $r_F(C) := C'$, that applies unique self-subsuming resolution for $C$ in $F$, if it is applicable, and defines $C' = C$ otherwise.
Based on $r$, we now apply the rule simultaneously for all clauses in $F$: we define $R(F) := \{r_F(C) \;|\;C \in F\}$.
This prevents rules initially applicable from deactivating one another.
By definition, $R$ is now a deterministic (and thus confluent) transformation.

We record that $R$ is still equivalence-preserving: in every application of $R$, we first perform all the resolution steps for the individual self-subsuming resolutions, followed by the subsumption steps.
We show that simultaneous self-subsumption is symmetry-preserving:
\begin{lemma} Simultaneous self-subsumption is symmetry-preserving. 
\end{lemma}
\begin{proof}
	Let $\varphi \in \Aut_{\syn}(F)$ and $F \xrightarrow{R} F'$ (slightly abusing notation).
	We prove $\Aut_{\syn}(F)\Sred_{\Lit(F')} \subseteq \Aut_{\syn}(F')$ holds.
	Let $C$ denote a clause where a literal $x$ can be removed through unique self-subsuming resolution.
	It immediately follows from syntactic symmetry, that the same holds true for $\varphi(C)$ and $\varphi(x)$.

	It remains to be shown that $\Aut_{\syn}(F) = \Aut_{\syn}(F)_{\{\Lit(F')\}}$.
	Again, note that if a literal $x$ is removed exhaustively from $F$ in $F'$, then by symmetry, the same holds true for $\varphi(x)$.
\end{proof}

\xparagraph{Learning clauses.} 
Next, we briefly consider learning clauses in a CDCL solver. The transformation adds clauses to $F$ that can be derived from $F$ using resolution.
It is easy to see this transformation is indeed equivalence-preserving, and thus Lemma~\ref{lem:variable_preserve_semantic} applies.

However, learning clauses is not symmetry-preserving. 
Consider $(x \vee a) \wedge (\n{x} \vee a) \wedge (x \vee b) \wedge (\n{x} \vee b)$.
In the formula, $a$ and $b$ are syntactically symmetrical.
But if we derive and add, e.g., the clause $(a)$, $a$ and $b$ are not symmetrical anymore.

We could ensure a symmetry-preserving transformation by, e.g., always either adding all symmetrical learned clauses, or only adding clauses on asymmetrical variables (analogous to the technique that we will apply to variable elimination later on in Section~\ref{sec:solution_alter}). 
Also, although impractical, adding learned clauses exhaustively is symmetry-preserving as well.

\subsection{Equisatisfiability-Preserving Transformations} \label{sec:solution_alter}
We now turn our attention to rules that may alter the set of solutions, i.e., transformations that are only \emph{equisatisfiability-preserving}. 
Since the reduced formulas change the set of satisfying and conflicting assignments, we need to adjust our expectations for the resulting symmetries. 
Indeed, Lemma~\ref{lem:variable_preserve_semantic} becomes incompatible.
This means a transformation might now indeed not be symmetry-lifting, and not weakly symmetry-preserving.

\xparagraph{Unit.} 
We analyze the unit literal rule. 
The unit literal rule requires that there is a unit clause $\{l\} \in F$ and consequently reduces the formula to $F[l \mapsto \top]$.

Since unit can lead to conflicts, we add a \emph{conflict rule}, such that if $\{\} \in F$, we can reduce $F$ to $F' = \{\{\}\}$.
In case of a conflict, we thus reduce the formula to a unique conflicting formula. 
We record that with the conflict rule, the rules are confluent \cite{DBLP:journals/jair/HeuleJLSB15}.
Indeed, we can now prove the rules to be both symmetry-preserving and symmetry-lifting.

\begin{lemma} Exhaustive unit and conflict is symmetry-preserving\label{lem:unit_exhaustive} and symmetry-lifting. \label{lem:unit_lifting}
\end{lemma}
\begin{proof}
\textit{(Symmetry-preserving.)} 
	Since unit and conflict are confluent, it suffices to show that unit and conflict are isomorphism-invariant to apply Lemma~\ref{lem:con_iso_implies_sympreserve}: whenever a literal $l$ is unit in $F$ and we reduce $F$ to $F[l \mapsto \top]$, then indeed, $\varphi(l)$ is unit in $\varphi(F)$ and we may reduce it to $\varphi(F)[\varphi(l) \mapsto \top] = \varphi(F[l \mapsto \top])$ for all $\varphi \in \SymL(\Lit(F))$.

\textit{(Symmetry-lifting.)} Let $F \xrightarrow{\Pi} F'$ (where $\Pi$ denotes unit and conflict).
We prove $\Aut_{\sem}(F') \subseteq \Aut_{\sem}(F)$. 
Let $\varphi \in \Aut_{\sem}(F')$ and $\sigma$ be a complete assignment of $F$.
Furthermore, let $\unit'(F)$ denote the unit literals assigned in $F'$. 
If $\sigma$ does not assign all literals in $\unit'(F)$ positively, $F[\sigma]$ is unsatisfiable and the permutations of $\varphi$ clearly have no effect on this. 
Hence, we may assume that all unit literals are assigned correctly in $\sigma$. 
Since this reduces the formula to $F'$, symmetries of $F'$ now apply.
\end{proof}

\xparagraph{Pure.}
The pure literal rule requires that there is a literal $l \in \Lit(F)$ such that for every clause $C \in F$ it holds that $\n{l} \notin C$. We may then assign $F[l \mapsto \top]$. Formally, we define 
\[\frac{l \in \Lit(F) \quad \forall C \in F: \n{l} \notin C}{F[l \mapsto \top]}.\]
We again consider the transformation that applies pure exhaustively:
\begin{lemma} \label{lem:pure_exhaustive} Exhaustive pure is symmetry-preserving.
\end{lemma}
\begin{proof}
	We apply Lemma~\ref{lem:con_iso_implies_sympreserve} to show the claim.
	Note that pure is confluent \cite{DBLP:journals/jair/HeuleJLSB15}.
	It thus suffices to show that pure is isomorphism-invariant, which follows readily: whenever a literal $l$ is pure in $F$ and we reduce $F$ to $F[l \mapsto \top]$, then indeed, $\varphi(l)$ pure in $\varphi(F)$ and we may reduce it to $\varphi(F)[\varphi(l) \mapsto \top] = \varphi(F[l \mapsto \top])$ for all $\varphi \in \SymL(\Lit(F))$.
\end{proof}
The question remains whether pure is symmetry-lifting.
The technique used to prove unit to be symmetry-lifting strongly depended on the fact that wrongly assigning a literal immediately leads to a conflict. 
For pure literals, this is simply not true: ``wrongly'' assigning a pure literal is of course not helpful towards proving a formula satisfiable, however, a formula might still be satisfiable. Moreover, say we have two pure literals $l_1$ and $l_2$. 
It might even be the case that assigning $\n{l_1}$ makes the formula unsatisfiable, while assigning $\n{l_2}$ still leaves the formula satisfiable. 
This means pure literals can indeed be distinct on a semantic level.
Indeed, pure is not symmetry-lifting: 
\begin{corollary} \label{lem:pure_symmetry_lifting} Exhaustive pure is not symmetry-lifting.
\end{corollary}
\begin{proof}
	Consider $F = (x \vee z) \wedge (\n{y} \vee z) \wedge (\n{y} \vee \n{z}) \wedge (y \vee b)  \wedge (y \vee a) \wedge (y \vee b) \wedge (\n{y} \vee \n{a} \vee \n{b}).$
	If we assign $F[x \mapsto \top] = (\n{y} \vee z) \wedge (\n{y} \vee \n{z}) \wedge (y \vee a) \wedge (y \vee b) \wedge (\n{y} \vee \n{a} \vee \n{b})$ ($x$ is pure), we get the syntactic symmetry $\varphi = (z\n{z})$. However, if we want $\varphi$ to be a semantic symmetry of $F$, the symmetry must also be valid in the case where we assign $x$ differently, i.e., where we have $F' := F[x \mapsto \bot] = (z) \wedge (y \vee a) \wedge (\n{y} \vee z) \wedge (\n{y} \vee \n{z}) \wedge (y \vee a) \wedge (y \vee b)  \wedge (\n{y} \vee \n{a} \vee \n{b})$.
	Clearly $F'$ is still satisfiable, but $\varphi$ is no semantic symmetry: $F'[z \mapsto \top]$ is satisfiable but $F'[z \mapsto \bot]$ is unsatisfiable.
\end{proof}

\xparagraph{Blocked Clause Elimination (BCE).} 
Let us now consider blocked clause elimination \cite{DBLP:conf/tacas/JarvisaloBH10}.
A literal $l \in C$ blocks $C$, if for every clause $C' \in F$ with $\n{l} \in C'$ it holds that $C \circ_l C'$ is a tautology.
A clause is called blocked whenever there is a literal that blocks it.
Blocked clause elimination exhaustively removes blocked literals from a formula.
Since blocked clause elimination is confluent \cite{DBLP:conf/tacas/JarvisaloBH10}, we can again consider the deterministic transformation that exhaustively applies blocked clause elimination.

We show that blocked clause elimination is not symmetry-lifting.
\begin{corollary} \label{lem:bce_symmetry_lifting} Exhaustive blocked clause elimination is not symmetry-lifting.
\end{corollary}
\begin{proof}
	We show the result in a similar fashion to Corollary~\ref{lem:pure_symmetry_lifting}.
	Consider $F = (x \vee z) \wedge (\n{y} \vee z) \wedge (\n{y} \vee \n{z}) \wedge (y \vee a) \wedge (y \vee b) \wedge (\n{a} \vee \n{b} \vee \n{c}) \wedge (c \vee \n{y}).$
	Again, we assign $F[x \mapsto \top] = (\n{y} \vee z) \wedge (\n{y} \vee \n{z}) \wedge (y \vee a) \wedge (y \vee b) \wedge (\n{a} \vee \n{b} \vee \n{c}) \wedge (c \vee \n{y})$ ($x$ is pure), yielding the syntactic symmetry $\varphi = (z\n{z})$. Also, note that BCE can not be applied further to $F[x \mapsto \top]$. 
	We consider $F[x \mapsto \bot] = (z) \wedge (\n{y} \vee z) \wedge (\n{y} \vee \n{z})\wedge (y \vee a) \wedge (y \vee b) \wedge (\n{a} \vee \n{b} \vee \n{c}) \wedge (c \vee \n{y})$, where $\varphi$ is no semantic symmetry: $F'[z \mapsto \top]$ is satisfiable but $F'[z \mapsto \bot]$ is unsatisfiable.
\end{proof}
However, it is indeed symmetry-preserving:
\begin{lemma} Exhaustive blocked clause elimination is symmetry-preserving.
\end{lemma}
\begin{proof} 
	We again apply Lemma~\ref{lem:con_iso_implies_sympreserve} to show the claim.
	Note that BCE is confluent \cite{DBLP:conf/tacas/JarvisaloBH10}.
	It thus suffices to show that BCE is isomorphism-invariant, which follows readily: whenever a literal $l \in C$ blocks $C$ in $F$, then $\varphi(l) \in \varphi(C)$ blocks $\varphi(C)$ in $\varphi(F)$ for all $\varphi \in \SymL(\Lit(F))$.
\end{proof}

\xparagraph{Bounded Variable Elimination (BVE).}
Next, we consider variable elimination \cite{DBLP:journals/jacm/DavisP60}, which is a crucial component of SAT preprocessors \cite{DBLP:conf/sat/EenB05}.
Let $C_x(F)$ and $C_{\n{x}}(F)$ denote the clauses containing $x$ (or $\n{x}$) of $F$. 
Variable elimination of a variable $x$ removes all clauses containing $x$, while adding $R_x(F) := \{C_1 \circ_x C_2 \;|\; C_1 \in C_x(F), C_2 \in C_{\n{x}}(F)\}$.
Overall, elimination of $x$ produces
$F' = F \cup R_x(F) \smallsetminus (C_x(F) \cup C_{\n{x}}(F)).$
Usually, variable elimination is only applied in a \emph{bounded} form (BVE).
This means that the rule is only applied whenever $F'$ is ``smaller'' than $F$. 
Different metrics for ``smaller'' are used.
One example is the number of clauses, i.e., simply checking whether $|F'| < |F|$ (see \cite{DBLP:conf/sat/EenB05}).

Let us first consider whether (exhaustive) BVE is symmetry-preserving.
Using the following example, we can show that this is not the case:
\begin{equation*}
\begin{split}
(x \vee a_1) \wedge (x \vee b_1) \wedge (\n{x} \vee c_1) \wedge (\n{x} \vee d_1) \wedge \\
(y \vee a_2) \wedge (y \vee b_2) \wedge (\n{y} \vee c_2) \wedge (\n{y} \vee d_2) \wedge \\
(x \vee A) \wedge (y \vee A) \wedge (\n{x} \vee A) \wedge (\n{y} \vee A) \wedge \\
(a_1 \vee A) \wedge (a_2 \vee A)  \wedge (b_1 \vee A) \wedge (b_2 \vee A) \wedge (c_1 \vee A) \wedge (c_2 \vee A)  \wedge (d_1 \vee A) \wedge (d_2 \vee A) \wedge (A) \wedge \\
(a_1 \vee z_1) \wedge (b_1 \vee z_1) \wedge (c_1 \vee \n{z_1}) \wedge (d_1 \vee \n{z_1}) \wedge (a_2 \vee z_2) \wedge (b_2 \vee z_2) \wedge (c_2 \vee \n{z_2}) \wedge (d_2 \vee \n{z_2}) \wedge\\
(\n{a_1} \vee \n{b_1} \vee \n{c_1}) \wedge (\n{b_1} \vee \n{c_1} \vee \n{d_1}) \wedge (\n{a_1} \vee \n{c_1} \vee \n{d_1}) \wedge \\
(\n{a_2} \vee \n{b_2} \vee \n{c_2}) \wedge (\n{b_2} \vee \n{c_2} \vee \n{d_2}) \wedge (\n{a_2} \vee \n{c_2} \vee \n{d_2})
\end{split}
\end{equation*}
Note that initially, $x$ and $y$ are symmetrical. We now eliminate $x$. For exhaustive BVE to be symmetry-preserving, in any exhaustive application of the rule $y$ must in turn also be removed.
\begin{equation*}
\begin{split}
(y \vee a_2) \wedge (y \vee b_2) \wedge (\n{y} \vee c_2) \wedge (\n{y} \vee d_2) \wedge\\
(a_1 \vee A) \wedge (a_2 \vee A)  \wedge (b_1 \vee A) \wedge (b_2 \vee A) \wedge (c_1 \vee A) \wedge (c_2 \vee A)  \wedge (d_1 \vee A) \wedge (d_2 \vee A) \wedge (A) \wedge\\
(y \vee A) \wedge (\n{y} \vee A) \wedge \\
(a_1 \vee z_1) \wedge (b_1 \vee z_1) \wedge (c_1 \vee \n{z_1}) \wedge (d_1 \vee \n{z_1}) \wedge (a_2 \vee z_2) \wedge (b_2 \vee z_2) \wedge (c_2 \vee \n{z_2}) \wedge (d_2 \vee \n{z_2}) \wedge\\
(\n{a_1} \vee \n{b_1} \vee \n{c_1}) \wedge (\n{b_1} \vee \n{c_1} \vee \n{d_1}) \wedge (\n{a_1} \vee \n{c_1} \vee \n{d_1}) \wedge \\
(\n{a_2} \vee \n{b_2} \vee \n{c_2}) \wedge (\n{b_2} \vee \n{c_2} \vee \n{d_2}) \wedge (\n{a_2} \vee \n{c_2} \vee \n{d_2}) \wedge \\
(a_1 \vee c_1) \wedge (a_1 \vee d_1) \wedge (b_1 \vee c_1) \wedge (b_1 \vee d_1)
\end{split}
\end{equation*}
The resulting equation does indeed have fewer clauses than before. Note that the above already, unsurprisingly, implies that non-exhaustive BVE is not guaranteed to be symmetry-preserving. Next, we choose to eliminate $A$. Since $A$ is pure, this again reduces the number of clauses and leads to the following formula.
\begin{equation*}
\begin{split}
(y \vee a_2) \wedge (y \vee b_2) \wedge (\n{y} \vee c_2) \wedge (\n{y} \vee d_2) \wedge \\ 
(a_1 \vee z_1) \wedge (b_1 \vee z_1) \wedge (c_1 \vee \n{z_1}) \wedge (d_1 \vee \n{z_1}) \wedge (a_2 \vee z_2) \wedge (b_2 \vee z_2) \wedge (c_2 \vee \n{z_2}) \wedge (d_2 \vee \n{z_2}) \wedge\\
(\n{a_1} \vee \n{b_1} \vee \n{c_1}) \wedge (\n{b_1} \vee \n{c_1} \vee \n{d_1}) \wedge (\n{a_1} \vee \n{c_1} \vee \n{d_1}) \wedge \\
(\n{a_2} \vee \n{b_2} \vee \n{c_2}) \wedge (\n{b_2} \vee \n{c_2} \vee \n{d_2}) \wedge (\n{a_2} \vee \n{c_2} \vee \n{d_2}) \wedge \\
(a_1 \vee c_1) \wedge (a_1 \vee d_1) \wedge (b_1 \vee c_1) \wedge (b_1 \vee d_1)
\end{split}
\end{equation*}
Finally, we reduce $z_1$, which produces no new clauses.
\begin{equation*}
\begin{split}
(y \vee a_2) \wedge (y \vee b_2) \wedge (\n{y} \vee c_2) \wedge (\n{y} \vee d_2) \wedge \\ 
(a_2 \vee z_2) \wedge (b_2 \vee z_2) \wedge (c_2 \vee \n{z_2}) \wedge (d_2 \vee \n{z_2}) \wedge\\
(\n{a_1} \vee \n{b_1} \vee \n{c_1}) \wedge (\n{b_1} \vee \n{c_1} \vee \n{d_1}) \wedge (\n{a_1} \vee \n{c_1} \vee \n{d_1}) \wedge \\
(\n{a_2} \vee \n{b_2} \vee \n{c_2}) \wedge (\n{b_2} \vee \n{c_2} \vee \n{d_2}) \wedge (\n{a_2} \vee \n{c_2} \vee \n{d_2}) \wedge \\
(a_1 \vee c_1) \wedge (a_1 \vee d_1) \wedge (b_1 \vee c_1) \wedge (b_1 \vee d_1)
\end{split}
\end{equation*}
Depending on whether tautologies can be discarded or not, variables in the cluster $\{a_1, b_1, c_1, d_1\}$ can be reduced further. Note that this cluster of variables is independent from the rest of the variables. However, in any case, using this variable order, we are now not able to eliminate $y$ (or any of the variables in $\{a_2, b_2, c_2, d_2, z_2\}$), which would be required to satisfy the symmetry-preserving property.

Using a similar example we can show that (exhaustive) BVE is also not symmetry-lifting:
\begin{equation*}
\begin{split}
(y \vee a) \wedge (y \vee b) \wedge (\n{y} \vee c) \wedge (\n{y} \vee d) \wedge \\ 
(a \vee z) \wedge (b \vee z) \wedge (c \vee \n{z}) \wedge (d \vee \n{z}) \wedge \\
(\n{a} \vee \n{b} \vee \n{c}) \wedge (\n{b} \vee \n{c} \vee \n{d}) \wedge (\n{a} \vee \n{c} \vee \n{d}) \wedge \\
(A \vee a)
\end{split}
\end{equation*}
Initially, $a$ and $b$ are not symmetrical due to the clause $(A \vee a)$. Eliminating $A$ only removes the clause $(A \vee a)$. Now, $a$ and $b$ are symmetrical, in particular there is the symmetry that only interchanges $a$ and $b$. However, considering the case where $[A \mapsto \bot]$, we can see that this is not a semantic symmetry of the original formula: $a$ is a unit literal and hence assigning $[a \mapsto \bot]$ makes the formula unsatisfiable, while assigning $[b \mapsto \bot]$ still leaves the formula satisfiable.

However, we can still show the following.
\begin{lemma} Let $F \xrightarrow{\text{VE}} F'$. It holds that $\Aut_{\syn}(F)\Sred_{\Lit(F')} \leq \Aut_{\syn}(F')$. \label{lem:ve_almost_sympre}
\end{lemma}
\begin{proof}

Let $F'$ denote the formula where the literals of $L$ were eliminated from $F$.
Let $\varphi \in \Aut_{\syn}(F)\Sred_{\Lit(F')}$. 
We prove that $\varphi \in \Aut_{\syn}(F')$.

Consider a clause $C \in F$ where $C \cap L = \emptyset$, i.e., that was not involved in variable elimination.
It follows immediately that $\varphi(C) \cap L  = \emptyset$.
Hence, $C \in F'$ and $\varphi(C) \in F'$.

Now consider $C \in F$ with $C \cap L = L' \neq \emptyset$.
It follows that $\varphi(C) \in F$ with $\varphi(C) \cap L = \varphi(L')$.
We use the fact that if multiple variables $v_1, \dots{}, v_m$ are eliminated, the order in which they are eliminated does not change the set of clauses produced.
Let $\{C_1', \dots{}, C_k'\} \subseteq F'$ denote the clauses that were resolved using $C$ when eliminating $L'$.
We prove that $\{\varphi(C_1'), \dots{}, \varphi(C_k')\}$ are the corresponding clauses created from $\varphi(C)$ when eliminating $\varphi(L') \subseteq L$. 
If $k = 0$, then by symmetry, the same follows for $\varphi(C)$.

We can see that in the first step, $C$ is resolved with $C_1, \dots{}, C_{k_1}$ each containing $v_1$ respectively.
We can mimic this step on $\varphi(C)$, resolving with $\varphi(C_1), \dots{}, \varphi(C_{k_1})$ each containing $\varphi(v_1)$ -- which all exist by symmetry.
In fact, we can mimic every step for $v_1, \dots{}, v_m$, using $\varphi(v_1), \dots{}, \varphi(v_m)$.
While this does not match the order used by variable elimination, we do know that $\{\varphi(v_1), \dots{}, \varphi(v_m)\} \subseteq L$, i.e., eventually we eliminate all involved variables.

The clauses involved in the resolution can then again introduce more literals involved in variable elimination, i.e., $(C_1' \cup \dots{} \cup C_k') \cap L = L'' \neq \emptyset$ (but of course $L'' \cap L' = \emptyset$).
Due to the setwise stabilizer, this however also happens in a symmetrical fashion for clauses produced from $C$ and $\varphi(C)$, i.e., the corresponding set for $\varphi(C)$ is $\varphi(L'')$.
Hence, we can repeat the argument recursively until no further literals of $L$ are introduced, and the claim follows.
\end{proof}

\textbf{Symmetric Variable Elimination.}
Lemma~\ref{lem:ve_almost_sympre} opens up the opportunity for \emph{symmetric variable elimination}.
If we can somehow ensure $\Aut_{\syn}(F)_{\{\Lit(F')\}} = \Aut_{\syn}(F)$, then the technique becomes symmetry-preserving.
The idea is that we can do so using any over-approximation of the orbit partition of $\Aut_{\syn}(F)$.

Let $\sigma$ denote the orbits of $\Aut_{\syn}(F)$, i.e., let $l_1, l_2 \in \Lit(F)$, then $\sigma(l_1) = \sigma(l_2)$ if and only if there exists a $\varphi \in \Aut_{\syn}(F)$ with $\varphi(l_1) = l_2$.
We can naturally lift this to variables, i.e., two variables $v_1, v_2$ are in the same orbit if either $\sigma(v_1) = \sigma(v_2)$ or $\sigma(v_1) = \sigma(\n{v_2})$.

A step of symmetric variable elimination eliminates a set of variables $V$, i.e., it performs multiple steps of variable elimination.
The crucial requirement is that for each $v \in V$ it additionally holds that $\sigma^{-1}(\sigma(v)) \subseteq V$, i.e., the entire orbit of $v$ must be contained in $V$.
This means that symmetric variable elimination is only allowed to eliminate unions of orbits simultaneously.
We can immediately conclude the following from Lemma~\ref{lem:ve_almost_sympre}:
\begin{corollary} Symmetric variable elimination is symmetry-preserving.
\end{corollary} 
Following our practical arguments, we should obviously not truly use the orbit partition, since this would defeat the purpose of using preprocessing to speed-up symmetry detection.
However, practical graph isomorphism solving is entirely based around fast over-approximations of the orbit partition.
The main subroutine of practical graph isomorphism solvers is \emph{color refinement} (or 1-dimensional Weisfeiler-Leman), which precisely delivers such a partition \cite{McKay201494}.

\section{A Symmetry-preserving Preprocessor} 
We implement and test a proof-of-concept, symmetry-preserving preprocessor.

\subsection{Implementation and Refined Graph Encoding} \label{sec:impl}
The main strategy is to apply the symmetry-preserving CNF simplifications discussed in the previous section.
Firstly, the preprocessor performs exhaustive unit, pure, subsumption, simultaneous self-subsumption and blocked clause elimination.
We implemented subsumption and simultaneous self-subsumption based on the techniques described for \textsc{SatELite} \cite{DBLP:conf/sat/EenB05}. 
Then, we perform bounded variable elimination, again following the heuristics of \textsc{SatELite}.
However, we only do this step for variables guaranteed to be asymmetrical.
Since we only do this step for asymmetrical variables, it becomes trivial to simultaneously guarantee symmetric variable elimination and bounded variable elimination.
We ensure variables to be asymmetrical by performing color refinement on the model graph of the reduced formula.
We employ the color refinement implementation of \textsc{dejavu} \cite{DBLP:conf/alenex/AndersS21}. 

There is one more technique we use to improve the graph encoding.
Before computing symmetries, we reduce the graph using the coloring computed by color refinement.
We simply remove all asymmetrical parts of the graph, i.e., all nodes which have a discrete color.
For readers familiar with graph isomorphism solvers, this strategy might be counter-intuitive.
Essentially, the strategy mimics precisely the first step of graph isomorphism solvers.
However, in the concrete implementations of solvers, many components are simply tied directly to the number of initial vertices (e.g., incidence-lists or the dense Schreier-Sims algorithm in \textsc{Traces} and \textsc{dejavu}). 
\subsection{Benchmark Setup} \label{sec:experiments_setup}
We perform an evaluation on the $400$ instances of the main track of the SAT competition 2021 \cite{satComp2021}, comparing unprocessed versus preprocessed instances. 
Our evaluation comprises two parts, covering the two main aspects of detecting symmetry: measuring computation time of symmetry detection and breaking algorithms, and measuring the symmetry detected in the instances.
The goal of our benchmarks is to provide a proof-of-concept that there can be a tangible benefit in synergizing SAT techniques with symmetry detection.

The code for the preprocessor, solvers and benchmarks are available at \cite{engagesweb}. 

\textbf{Computation Time.} We test symmetry detection and breaking algorithms on unprocessed versus preprocessed instances. 
The goal is to find out whether symmetry-preserving formula transformations provide a benefit to algorithms concerning symmetry. 
We test three different graph isomorphism solvers, i.e., \saucy{} \cite{DBLP:conf/dac/DargaLSM04, DBLP:conf/dac/DargaSM08}, \Traces{} \cite{McKay201494, DBLP:journals/corr/abs-0804-4881} and \dejavu{} \cite{DBLP:conf/alenex/AndersS21, DBLP:conf/esa/AndersS21} (4 threads, $<1\%$ probability of missing a generator).
While all state-of-the-art solvers are based on the individualization-refinement framework, the tested solvers indeed represent distinct approaches to symmetry detection:	
first of all, \saucy{} was specifically designed to be used on graphs stemming from CNF formulas \cite{DBLP:conf/dac/DargaLSM04}.
	\Traces{} is well-known to be a highly-competitive general purpose GI solver. It uses a BFS approach and excels on hard combinatorial instances \cite{McKay201494}, while still featuring low-degree techniques typically deemed useful for CNF formulas \cite{DBLP:conf/dac/DargaLSM04, DBLP:conf/dac/DargaSM08} (e.g., unit literals always have degree $1$).
	\dejavu{} is a recent general-purpose GI solver. 
	It employs novel randomized search strategies \cite{DBLP:conf/alenex/AndersS21} and features parallelization \cite{DBLP:conf/esa/AndersS21}.
	It has no special low-degree vertex or CNF techniques.
	Furthermore, we measure running time of the symmetry breaker \textsc{BreakID} on the CNF formulas (here, we do not use the special graph encoding of Section~\ref{sec:impl}).
	Internally, \textsc{BreakID} uses \saucy{}.
	
	Note that even for the ``unprocessed'' instances, we first removed all redundant clauses and literals from the formula.
	The timeout is $100$ seconds, analogous to the evaluation of \cite{DBLP:conf/sat/Devriendt0BD16}.
	A solver running out of memory also counts as a timeout. 

	Our working assumption is that most of the proposed preprocessing (or inprocessing) techniques are applied anyway in competitive SAT solvers (see e.g., \cite{BiereFazekasFleuryHeisinger-SAT-Competition-2020-solvers, DBLP:conf/tacas/OsamaWB21, DBLP:conf/sat/SoosNC09}).
	Solvers such as \textsc{cryptominisat} \cite{DBLP:conf/sat/SoosNC09} even enable the user to re-schedule preprocessing and symmetry techniques.
	Our proposed setup can essentially be realized by carefully scheduling a selection of techniques in an existing preprocessor before symmetry detection, instead of implementing from scratch -- with the exception of our adapted symmetry-preserving transformations.
	But here, the only technique we apply that is unconventional in terms of SAT preprocessors is color refinement.

	We thus record and account for the computation time of color refinement separately.
	In other areas, our preprocessor is indeed poorly optimized compared to state-of-the-art implementations, and we will disregard its computation time.

	\textbf{Detected Symmetry.} We analyze how much symmetry is detected both before and after applying preprocessing. 
	To be more precise, we compute the following metrics.
	Let $F$ be the original formula and $F^*$ the preprocessed one.
	First of all, we check for \emph{reducible symmetry} that acts exclusively on literals removable through preprocessing, i.e., which we measure in terms of \[\frac{|\Aut_{\syn}(F)|}{|{\Aut_{\syn}(F)}_{(\Lit(F^*))}|}.\]
	The expression describes the number of symmetries of $F$ which pointwise-stabilize all literals of $F^*$, i.e., which solely act on the literals $\Lit(F) \smallsetminus \Lit(F^*)$. 

	Furthermore, we quantify the amount of \emph{hidden symmetry} detectable in the preprocessed but not the unprocessed instance, i.e., \[\frac{|\Aut_{\syn}(F^*)| \cdot |{\Aut_{\syn}(F)}_{(\Lit(F^*))}|}{|\Aut_{\syn}(F)|}.\]

\subsection{Benchmark Discussion}
\begin{figure}
	\centering
	\begin{subfigure}{0.24\textwidth}
	\scalebox{0.4}{
\begin{tikzpicture}

\definecolor{color0}{rgb}{0.12156862745098,0.466666666666667,0.705882352941177}

\begin{axis}[
log basis x={10},
log basis y={10},
tick align=outside,
tick pos=left,
x grid style={white!69.0196078431373!black},
xmajorgrids,
xmin=0.562341325190349, xmax=177827.941003892,
xmode=log,
xtick style={color=black},
y grid style={white!69.0196078431373!black},
ymajorgrids,
ymin=3.54813389233575e-05, ymax=281838.293126445,
ymode=log,
ytick style={color=black},
ylabel={time on preprocessed},
xlabel={time on unprocessed}
]
]
\addplot [draw=black, fill=black, mark=x, only marks]
table{%
x  y
58.5034 16.3831
5.00719 0.189569
44551.8 1897.48
142.622 0.02518
2165.31 0.0001
43.7667 12.4223
21.5338 0.0001
6167.62 0.0001
1935.02 1465.9
186.932 0.013604
3.94726 0.008358
2.10938 0.0001
58.2172 42.8927
165.84 0.015225
1262.1 1165.65
4650.36 0.0001
262.309 0.013754
100000 100000
3326.18 10.7297
5.1391 4.59821
16.7609 0.0001
3.85779 0.006084
2996.91 1189.19
484.935 345.72
7899.19 3436.93
616.966 293.59
17401 15840.9
125.637 0.0001
1319.59 2.85908
2.27947 0.008745
4730.2 13.5879
7384.96 1539.7
6134.94 1313.53
1441.14 555.525
2064.28 704.55
2437.88 1241.14
235.102 87.4614
100000 100000
7.55747 0.0001
2.05894 2.17047
1.51166 0.005771
21.7248 0.0001
72.2583 0.0001
111.408 0.015127
57543.5 26710.3
3123.43 2888.71
1796.26 1796.78
1069.79 3.19779
129.636 58.5543
3.43288 3.01515
359.332 473.897
1198.98 9.69224
3312.83 269.288
1642.65 667.831
140.021 73.2952
172.956 0.013672
77417.9 23377.6
278.708 272.853
3664.18 323.31
439.934 221.189
253.758 238.907
2.5747 0.0001
128.687 55.5775
226.605 232.485
3113.33 906.419
3599.2 11.2182
1617.71 567.201
3083.46 9.25927
10692.8 4265.52
2772.33 1063.77
448.271 126.499
8091.77 0.0001
20.8292 0.015727
1432.09 37.3895
121.805 57.7471
100000 100000
6383.23 0.0001
4037.37 11.9526
2.40685 0.0001
14.3693 14.3517
76.1522 0.013754
179.635 93.4669
118.326 64.8717
1168.35 436.472
1564.78 449.578
7378.45 2612.98
10214.2 41.1531
10.562 0.0001
70.5183 16.2757
10129.1 72.5651
384.555 13.8189
17.2109 1.94849
239.782 0.0001
7.41651 0.0001
6.94221 0.0001
3473.81 848.233
17.8736 12.893
19.3398 0.0001
1110.75 15.0183
484.158 0.0001
2633.97 8.91461
2915.73 9.01294
3.38196 0.00577
40.5687 30.9691
14706.6 14233.5
100000 100000
22.5352 0.014015
5922.4 0.0001
3.0045 3.46035
47.2318 14.2955
35.8885 0.01444
1200.37 1189.94
224.403 142.722
23.1717 0.0001
442.563 105.252
1527.64 0.0001
232.131 61.5419
291.927 282.563
2450.57 1042.9
4056.48 903.265
14788.3 13567.5
33.2946 0.013522
1393.06 769.808
100000 100000
3431.2 19.9484
4285.88 981.207
2.66089 0.0001
185.367 267.735
155.668 64.5652
8159.15 3794.59
2.0381 0.009405
2.51679 0.0001
100000 100000
49153.9 31882.8
9617.28 3700.95
3775.53 0.0001
1.13509 1.88766
245.58 0.976403
31.4819 0.014139
1.22802 0.0001
12.8574 4.63959
5129.76 10.2728
131.008 61.6303
3243.15 876.033
2452.8 8.40402
100000 100000
5150.66 10.8635
3.40393 0.006025
2.27549 0.0001
8077.17 10497.8
5351 14.0433
3132.32 690.662
3.05827 0.006317
3490.52 635.681
1246.02 438.738
128.785 0.01678
208.399 184.49
3995.71 9.27901
40898.8 11028.8
1666.47 598.406
94.7667 25.668
7419.98 0.0001
120.07 40.1403
16.8101 0.475316
69.6699 69.6159
1413.6 570.658
109.441 0.015435
816.423 805.193
18.6823 17.0532
463.612 553.125
119.588 108.168
4.22034 0.006278
1.45179 1.60426
12.8403 0.0001
1614.54 935.797
338.741 0.013573
35019.2 10271
100000 100000
1.7406 2.34904
6669.48 0.0001
352.609 161.633
1541.64 5.56882
100000 100000
5287.42 335.111
10153 0.0001
10392.8 7989.08
185.074 182.795
1.55368 0.0001
109.968 53.0166
1441.99 1508.31
4.517 0.006163
100000 100000
10809.1 0.0001
139.968 182.297
2.06346 1.77123
19980.1 18224
4764.29 7.5667
3.99694 0.0001
3754.8 1472.82
};
\draw[line width=4pt, green!30] (axis cs:0.1,0.0001) -- (axis cs:1000000,0.0001);
\addplot [semithick, color0]
table {%
1 1
100000 100000
};
\end{axis}

\end{tikzpicture}}
	\caption{\saucy{}.}
	\end{subfigure}
	\begin{subfigure}{0.24\textwidth}
		\scalebox{0.4}{
\begin{tikzpicture}

\definecolor{color0}{rgb}{0.12156862745098,0.466666666666667,0.705882352941177}

\begin{axis}[
log basis x={10},
log basis y={10},
tick align=outside,
tick pos=left,
x grid style={white!69.0196078431373!black},
xmajorgrids,
xmin=0.562341325190349, xmax=177827.941003892,
xmode=log,
xtick style={color=black},
y grid style={white!69.0196078431373!black},
ymajorgrids,
ymin=3.54813389233575e-05, ymax=281838.293126445,
ymode=log,
ytick style={color=black},
ylabel={time on preprocessed},
xlabel={time on unprocessed}
]
\addplot [draw=black, fill=black, mark=x, only marks]
table{%
x  y
118.117 46.6465
39.5592 0.719361
46512.4 781.126
12464 24.7176
100000 0.0001
87.2475 26.54
69.3318 0.0001
100000 0.0001
4087.44 3021.63
585.529 26.445
9.47085 0.274993
19.3133 0.0001
98.2936 81.1652
188.352 14.4501
1080.6 754.789
100000 0.0001
301.502 27.4938
100000 100000
100000 88.485
5.95397 27.9457
142.201 0.0001
10.1907 0.2417
4467.61 4905.61
757.959 742.304
3280.79 1022.77
655.302 476.878
5996.38 5810.33
852.108 0.0001
2981.11 17.8902
5.57047 0.276509
100000 115.054
2353.18 567.806
3269.38 557.627
10010 5240.72
792.08 371.587
6573.6 1764.59
894.156 306.365
77553.6 100000
35.0064 0.0001
1.49137 2.09686
3.71241 0.177851
74.8543 0.0001
308.958 0.0001
100000 26.3172
41116.4 18771.2
1402.05 1167.03
1004.74 1081.87
2438.16 8.28093
2055.03 290.246
1.92037 1.74367
3.65716 3.59276
4075.82 37.8395
1618.46 106.657
6145.29 3393.17
64.401 311.194
100000 0.313987
100000 16353.4
127.659 138.771
1638.91 121.31
712.804 518.069
109.257 141.934
30.7003 0.0001
466.372 224.163
509.113 287.668
1199.59 894.49
100000 90.0197
10634.4 5296.4
100000 85.983
3802.76 1797
2578.73 3354.08
211.303 91.2556
100000 0.0001
100000 0.289301
5318.41 172.949
2164.25 285.998
100000 100000
100000 0.0001
100000 95.5373
5.38113 0.0001
26.4964 15.4603
100000 0.387642
87.895 356.739
67.7627 273.783
4913.06 946.266
2973.43 1218.47
2618.15 1012.06
100000 339.371
29.7265 0.0001
307.734 45.3719
100000 615.768
1695.07 69.3439
34.2055 3.85889
16763.6 0.0001
30.6294 0.0001
34.6682 0.0001
1300.64 244.011
32.07 31.6375
63.3842 0.0001
3979.51 62.4625
8468.33 0.0001
100000 69.4302
100000 73.9515
8.67325 0.255998
68.7755 58.6381
5752.25 5116.13
100000 100000
100000 13.8587
100000 0.0001
3.73165 2.36289
120.922 54.422
100000 21.5265
899.233 973.566
103.251 144.004
124.327 0.0001
221.883 52.0603
100000 0.0001
160.831 64.9444
783.485 394.97
4578.5 3603.39
1592 319.721
7100.9 7174.36
100000 27.8209
776.466 376.065
100000 100000
13145.1 133.728
2251.42 2034.34
5.95134 0.0001
41.8018 56.2841
671.786 291.54
3211.52 1580.03
5.81489 0.273387
5.65576 0.0001
100000 100000
100000 100000
23828.2 100000
100000 0.0001
1.60656 4.33917
385.056 2.82698
100000 0.371612
3.99284 0.0001
27.426 10.552
100000 88.4692
226.065 152.621
2787.39 1864.17
100000 67.6233
100000 55709.1
100000 87.0214
14.5455 0.462796
5.0738 0.0001
1999.84 2815.7
100000 115.8
1057.09 402.858
21.0945 24.2703
13986.6 6320.63
7699.06 3939.41
100000 0.398084
108.354 126.007
100000 75.4225
100000 7151.37
10206.3 5134.01
377.061 69.8634
100000 0.0001
477.339 121.496
53.1696 1.31147
13.5872 13.8511
11519.3 6084.03
123.125 0.193556
644.35 526.636
27.5302 30.9225
2730.85 2982.02
65.2984 69.0714
11.5551 0.265835
4.55903 3.58575
98.4148 0.0001
4851.68 2167.86
100000 0.353522
12912.1 5770.38
100000 100000
21.0366 20.6499
100000 0.0001
158.648 641.04
3831.46 17.6826
100000 100000
5736.43 213.786
100000 0.0001
4295.96 2645.63
112.387 107.095
5.00025 0.0001
235.643 130.679
1251.17 886.295
11.5349 0.284837
100000 100000
100000 0.0001
143.234 167.262
41.1043 45.2081
6780.87 7596.54
11016 25.3763
35.5003 0.0001
9996.12 11179.9
};
\draw[line width=4pt, green!30] (axis cs:0.1,0.0001) -- (axis cs:1000000,0.0001);
\addplot [semithick, color0]
table {%
1 1
100000 100000
};
\end{axis}

\end{tikzpicture}}
	\caption{\dejavu{}.}
	\end{subfigure}
	\begin{subfigure}{0.24\textwidth}
		\scalebox{0.4}{
\begin{tikzpicture}

\definecolor{color0}{rgb}{0.12156862745098,0.466666666666667,0.705882352941177}

\begin{axis}[
log basis x={10},
log basis y={10},
tick align=outside,
tick pos=left,
x grid style={white!69.0196078431373!black},
xmajorgrids,
xmin=0.562341325190349, xmax=177827.941003892,
xmode=log,
xtick style={color=black},
y grid style={white!69.0196078431373!black},
ymajorgrids,
ymin=3.54813389233575e-05, ymax=281838.293126445,
ymode=log,
ytick style={color=black},
ylabel={time on preprocessed},
xlabel={time on unprocessed}
]
\addplot [draw=black, fill=black, mark=x, only marks]
table{%
x  y
99.4642 27.6565
9.11751 0.275418
100000 20775
73677.4 0.054646
3552.54 0.0001
89.1566 20.3444
36.2664 0.0001
9190.75 0.0001
5397.75 2077.52
298.273 0.03782
7.844 0.025763
3.89334 0.0001
76.7424 46.4668
331.222 0.033677
1253.03 1020.53
7311.28 0.0001
456.565 0.032095
100000 100000
100000 11.3775
78.5918 75.8072
32.4768 0.0001
7.96838 0.027094
5532.26 4622.4
541.36 477.126
8823.89 3439.46
545.054 310.546
8543.54 7616.59
370.379 0.0001
2639.14 4.10299
4.14737 0.02855
100000 15.3812
13096.1 2367.16
6179.58 1898.16
33990 3780.56
3135.27 1681.98
13742 1203.22
551.325 211.974
100000 100000
49.4791 0.0001
2.24323 2.08476
2.8491 0.03332
37.5317 0.0001
161.507 0.0001
70692.7 0.046682
100000 50836.6
2086.29 1392.82
1394.44 1207.15
2308.87 4.35592
2687.74 199.422
2.12913 2.63447
37.7699 34.1994
2593.16 14.7901
5750.62 232.297
2536.56 2256.7
335.809 22626.5
100000 0.035662
100000 100000
824.623 764.694
4719.36 396.221
546.9 276.74
507.929 2094.45
4.64553 0.0001
255.677 124.379
343.341 267.068
878.754 3150.95
100000 12.2491
37201 3960.93
100000 9.98296
18152.6 4782.36
2788.05 2427.27
1055.66 533.195
12285.9 0.0001
3529.23 0.035646
2885.69 61.0499
2183.67 217.171
100000 100000
10170.6 0.0001
100000 13.376
4.52026 0.0001
14.2082 14.2307
39826.6 0.040324
362.709 29231
226.237 30512.6
1679.34 572.517
2039.96 589.487
5929.72 2171.06
20600.3 92.2179
20.8061 0.0001
358.295 165.637
20795.6 176.242
838.996 22.7844
31.0714 2.39174
453.376 0.0001
33.9087 0.0001
59.1656 0.0001
3012.62 698.859
21.5321 17.3199
33.0911 0.0001
2423.83 25.4606
848.742 0.0001
100000 9.74802
100000 9.67431
7.13098 0.027993
51.3998 32.7302
8280.33 8564.9
100000 100000
49.7746 0.03953
9203.01 0.0001
4.95533 4.01252
87.3944 22.5954
13525.4 0.036604
1130.23 920.287
140.656 144.382
38.6204 0.0001
7949.13 568.788
2568.52 0.0001
427.456 167.007
455.798 379.618
5888.73 3602.5
4019.89 627.143
8516.54 7392.41
10794.4 0.042374
1779.42 892.816
100000 100000
7016.89 39.6229
2222.68 1837.11
5.06711 0.0001
469.767 478.605
1521.52 207.678
9384.71 4568.09
4.28921 0.028795
4.79625 0.0001
100000 100000
35224.3 36170.8
12494.8 5813.87
6296.35 0.0001
2.01567 2.17806
333.97 2.03744
4972.69 0.03915
2.18855 0.0001
31.2363 13.5764
100000 10.8318
197.787 114.32
3159.98 776.801
100000 8.83848
100000 100000
100000 11.9561
7.0151 0.029945
4.26931 0.0001
100000 100000
100000 14.3384
4382.67 2673.32
6.19692 0.027845
5296.01 1331.75
23042.8 2918.61
100000 0.038308
687.154 695.674
100000 10.1777
100000 100000
39814.2 3464.48
1034.5 242.297
12377.2 0.0001
305.03 79.245
29.425 0.607458
124.675 124.995
8160.1 4086.92
182.733 0.035882
830.314 795.654
18.6863 21.4495
5313.37 4734.91
360.161 356.479
8.97755 0.024498
2.79919 1.85255
23.9452 0.0001
2477.55 842.448
100000 0.035634
100000 100000
100000 100000
3.18114 3.06321
10563.5 0.0001
808.069 43527.3
3103.43 9.92677
100000 100000
7673.6 337.798
14161.1 0.0001
14561.4 9473.4
600.954 724.728
3.0407 0.0001
172.974 87.0524
1455.31 1309.17
11.034 0.029332
100000 100000
16819.4 0.0001
412.954 741.229
1.89399 1.5411
10859.2 9513.34
8845.81 12.65
8.31924 0.0001
5065.27 2313.23
};
\draw[line width=4pt, green!30] (axis cs:0.1,0.0001) -- (axis cs:1000000,0.0001);
\addplot [semithick, color0]
table {%
1 1
100000 100000
};
\end{axis}

\end{tikzpicture}}
	\caption{\Traces{}.}
	\end{subfigure}
	\begin{subfigure}{0.24\textwidth}
		\scalebox{0.4}{
\begin{tikzpicture}

\definecolor{color0}{rgb}{0.12156862745098,0.466666666666667,0.705882352941177}

\begin{axis}[
log basis x={10},
log basis y={10},
tick align=outside,
tick pos=left,
x grid style={white!69.0196078431373!black},
xmajorgrids,
xmin=0.562341325190349, xmax=177827.941003892,
xmode=log,
xtick style={color=black},
y grid style={white!69.0196078431373!black},
ymajorgrids,
ymin=0.562341325190349, ymax=177827.941003892,
ymode=log,
ytick style={color=black},
ylabel={time on preprocessed},
xlabel={time on unprocessed}
]
]
\addplot [draw=black, fill=black, mark=x, only marks]
table{%
x  y
607.361166999908 610.980500000096
53.7217790001705 35.392408000007
100000 9971.08663599965
5889.65538100001 988.784222000049
19737.4744220006 16332.1341430001
450.136984999517 447.661051000068
262.201096000354 24.075093000647
40653.3323249996 37596.0766440003
14106.5866919998 6545.60082399985
3462.98118400045 2.52966700008983
63.5131609997188 50.7317669998884
41.219285999432 38.6657999997624
324.140537999483 304.019776000132
2580.61681300023 1379.44616799996
4032.26993800035 4652.02298700024
33420.1604939999 32682.1283170002
4137.80855999994 2233.20665600022
100000 100000
100000 17882.0477729996
35.1140310003757 30.3799579996848
307.310488999974 262.449614999241
66.5424710005027 54.3394070000431
10438.7891380002 7863.40281100001
2684.48021699987 2741.8278819996
13136.9469169995 5652.9184609999
2094.9178740002 1847.73344199948
39720.7146159999 40463.7597160008
2581.72722299969 1548.49379799998
9069.24081899979 3404.0622499997
36.5734599999996 28.5972689998744
100000 28259.0305620006
12795.288192 2585.71172500069
11647.5285839997 2994.63318500057
16656.9485029995 17752.9264619998
6982.67841500001 2992.39557500005
20137.2146389995 5641.96636100041
2877.56063599954 2536.89436100012
100000 100000
153.187828998853 105.709801000557
7.69140000011248 7.87974400009261
26.4924710008927 22.8839430001244
349.605189001522 30.8093759995245
1130.99839499955 786.841072000243
50337.9843840012 1358.06888600018
100000 90100.0735329999
10097.875391999 9579.59528399988
7672.35089099995 8123.27510700015
10054.6167350003 4080.12844199948
1973.73479800081 1746.84455199895
11.1301079996338 11.7290210000647
1100.61381199921 456.252183999823
11277.8736849996 4334.52527500049
8274.60531999895 663.448180999694
11418.900961 10278.8457840015
587.385713000913 3810.52208799883
99015.7043930012 1981.03943600108
100000 100000
1033.71182800038 971.228444001099
7548.6573299986 732.341961000202
3430.71928599966 3045.56844800027
816.063218000636 727.225062999423
48.4198080011993 43.0774609994842
1324.41376700081 1269.5373610004
978.777148000518 996.284332999494
15659.1675029995 9549.94873000032
100000 22038.3921039993
19548.9870689998 19856.2076790004
100000 18095.8162349998
18680.7463919995 9570.73486800073
9606.7645940002 9254.87772499946
1920.66416799935 764.775398998609
54555.3600439998 53018.2209089999
4184.21894899984 223.092112000813
16863.7865309993 17079.5514759993
1951.76856699982 1732.0661800004
100000 100000
51266.2995680002 48979.8066210005
100000 25342.5573380009
52.7514749992406 46.5593379994971
57.6048010007071 57.6052449996496
29788.5955959991 851.340757000798
670.049573998767 4666.03123599998
386.883518000104 2432.30165900059
4009.35079600094 2178.47133800024
6219.52029100066 3746.08049100061
13368.7946069986 6346.00238399798
61238.5326650001 28549.0891609988
143.839843000023 105.334706000576
1143.85950699943 677.621076996729
61788.5952680008 28704.9328389985
5301.86278100155 4959.54210299897
231.496905002132 241.417470002489
4613.12562000239 2592.61068299747
128.427016999922 98.9116140008264
131.333110002743 93.1084889998601
6253.12910299908 1400.13766500124
124.872363001487 122.803975998977
261.827158999949 25.1274219990592
9425.27881299975 3591.07941800175
5101.96873099994 4746.85688700265
100000 15841.7795759997
100000 17104.0808090002
61.8071499993675 49.2847249988699
253.833101000055 243.514164001681
39578.785978003 41718.5194500016
100000 100000
31612.0407849994 197.8987399998
46405.4568989995 43149.8590449992
35.5416729980789 16.8480299980729
559.545545002038 559.544375999394
9044.56116699657 405.012273997272
4930.05742199966 4945.51453299937
668.843076000485 359.845502000098
307.623500997579 28.3106320021034
748.527938001644 277.021324000088
11518.413247999 10729.3969319981
910.770985999989 317.988586000865
1165.69239399905 1126.76611300049
10029.8419379978 9204.15767799932
6725.81286000059 1307.92268900041
46011.4629609998 49161.5986270008
8610.28374299713 388.870441001927
2753.01559299987 1549.57994299912
100000 100000
23725.0786210025 7371.79774700053
12537.5905219989 8327.68205100001
56.9504630002484 56.3456529998803
603.446863999125 510.785669001052
1862.49064300137 1734.95311299848
17400.2590630007 10021.3937769986
40.4415819975839 32.2625929984497
53.6678989992652 50.4630010000255
100000 100000
100000 100000
76443.0662510022 78165.1686719997
30831.7758260018 29468.1987840013
18.0000729997118 17.6852610020433
2199.04505799786 1955.90716900188
7982.25207400174 363.236568999127
19.61012099855 14.6632139985741
181.389729001239 158.812261001003
100000 22205.4447019982
835.225761999027 837.803750997409
12772.7458430018 5278.76579800068
100000 15408.8823040001
100000 100000
100000 21397.0819989991
59.5205680001527 49.0620370001125
47.8803500009235 43.1383680006547
69207.1343149983 100000
100000 26213.4071719993
12186.8776869997 3719.41810900171
50.6729150001775 41.1728790022607
20359.1309469994 19792.4475420004
13663.9132029995 13989.3361110007
60635.724081003 1335.46081900204
502.298326999153 500.713269000698
100000 17839.9211299984
100000 100000
22086.8279239985 24776.3493070015
1461.06004399917 1102.75295199972
55823.6037089991 54683.4021230025
1790.12515100112 1563.38703499932
183.643938999012 119.218226001976
193.913301998691 189.942835997499
19455.4324759993 20456.8130500011
1801.85908600106 987.006814000779
4568.73655100208 4766.63660500344
74.0375759996823 70.0349130020186
5909.68174099908 7618.82687300022
382.382734002022 430.338702000881
85.4099960015446 67.3289020014636
25.8024810027564 15.912663002382
264.322861999972 227.819607996935
5917.44467899844 3512.5672159993
100000 3964.83482400072
60866.3845720002 25942.4966110018
100000 100000
27.1204180025961 19.8885299978429
41721.8878959975 40103.4230299992
1030.09576399927 7065.00013500045
10237.6213219977 4168.38787400047
100000 100000
10104.164411001 638.658053001564
22127.179213996 15756.062256005
656.860903000052 647.189777999301
30.5756520028808 15.0125319996732
821.708061994286 813.504555000691
6047.41320299945 8141.72447900637
90.4006499986281 72.5837429999956
100000 100000
460.224187998392 464.149642997654
8.53462899976876 9.24462100374512
70243.1043619945 79598.0655090025
32711.9974839952 10420.4039300021
97.0663390035043 86.432075004268
26938.3635569975 25736.0401169935
};
\addplot [semithick, color0]
table {%
1 1
100000 100000
};
\end{axis}

\end{tikzpicture}}
	\caption{\textsc{BreakID}.} \label{fig:breakid}
	\end{subfigure}
	\caption{Benchmarks for symmetry detection and breaking algorithms comparing unprocessed versus preprocessed SAT competition 2021 instances. Times are given in milliseconds. Instances that were processed to empty graphs (green bar) are shown separately.} \label{fig:sym_detect_times}
\end{figure}
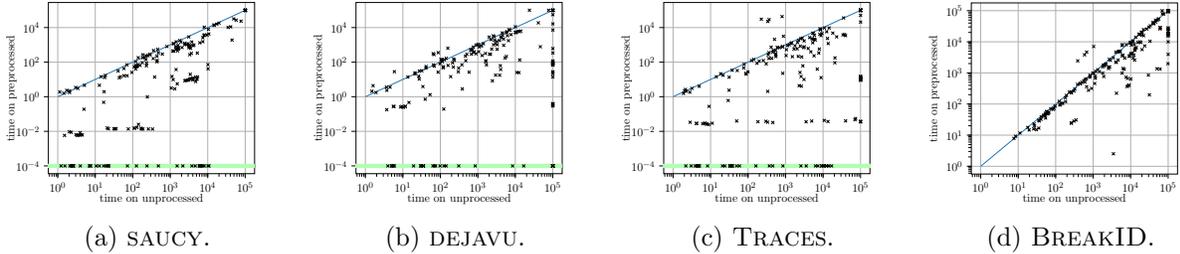

\begin{table}
	\centering
	\begin{tabular}{|l||c|c||c|c|}\hline
		   & \multicolumn{2}{c||}{unprocessed} & \multicolumn{2}{c|}{preprocessed}\\
	Solver & \multicolumn{1}{c}{\#finished} & avg.~(finished) & \multicolumn{1}{c}{\#finished} & avg.~(finished)\\\hline
	\saucy{}  & $\mathbf{189}$ & $3.67s$ & $\mathbf{189}$ & $\mathbf{1.28s}$\\\hline
	\dejavu{} & $151$ & $3.38s$ & $\mathbf{188}$ & $\mathbf{1.25s}$\\\hline
	\Traces{} & $168$ & $5.03s$ & $\mathbf{185}$ & $\mathbf{2.04s}$\\\hline
	\textsc{BreakID} & $170$ & $10.61s$ & $\mathbf{185}$ & $\mathbf{8.46s}$\\\hline	
\end{tabular}
	\caption{Benchmarks for symmetry detection and breaking algorithms comparing unprocessed versus preprocessed SAT competition 2021 instances ($199$ instances). Time taken for preprocessing is not included. The average time given is only for instances that finished within the $100s$ timeout.} \label{tab:sym_detect_times}
\end{table}

For the sake of clarity, in the entire evaluation, we exclude instances that were asymmetrical in \emph{both} the unprocessed and preprocessed instance. 
This was the case for $192$ instances. Note that asymmetrical instances in the set are solved quickly by the graph isomorphism solvers and there are only negligible differences between solvers (e.g. \saucy{} took $26s$ for all asymmetrical instances).
Furthermore, our preprocessor failed to process $9$ instances, either because it ran out of memory, or took more than $1000$ seconds.
Thus, in the following, all results are stated for the remaining $199$ instances.

\textbf{Computation Time.} 
Firstly, the preprocessor itself took a total time of $8946$ seconds ($45s$ average).
The color refinement algorithm for all the instances took $255$ seconds ($1.28s$ average). 
The results for the different symmetry detection and breaking algorithms are summarized in Table~\ref{tab:sym_detect_times}. 

We observe a speed-up for all tested tools across the benchmark suite (see Figure~\ref{fig:sym_detect_times}): we observe the same number or fewer timeouts for all tools. The overall average symmetry detection times in (on instances that finished) are between $2.47$ and $3.38$ times faster ($1.25$ for \textsc{BreakID}). When the cost for color refinement is accounted for, instances are still solved on average between $1.36$ and $1.51$ times faster ($1.09$ for \textsc{BreakID}) -- but especially when additionally considering the reduced number of timeouts, the overhead cost for color refinement is easily amortized.

For the general-purpose graph isomorphism solvers \Traces{} and \dejavu{}, the difference between unprocessed and preprocessed instances seems dramatic (a factor of $2.2$ and $4.3$ fewer timeouts, respectively).
In particular, \dejavu{} manages to close most of the performance gap to \saucy{}.
This is quite surprising, since \dejavu{} features no CNF-oriented techniques at all.
We record that on the preprocessed instances, \dejavu{} even manages to uniquely solve one instance that \saucy{} can not within the timeout -- albeit while \saucy{} does so for two other instances.  
Overall, on the preprocessed instances the performance of solvers is more comparable and thus seems less dependent on CNF-oriented optimizations. 

For the symmetry-breaking tool \textsc{BreakID}, which internally uses \saucy{}, the difference in the number of timeouts is also substantial -- in particular much more so than for \saucy{} individually ($2.07$ times fewer timeouts for \textsc{BreakID} versus the same number of timeouts for \saucy{}). 
Looking more closely at the results reveals that the removal of reducible symmetry speeds up \textsc{BreakID}:
while \saucy{} is often able to handle a lot of symmetry more efficiently through the CNF-oriented techniques, the computation time of \textsc{BreakID} depends directly on the number of generators detected. 

We record that in our benchmarks, the computation time taken up by our rudimentary preprocessor is not amortized by the time saved during symmetry detection. For a discussion as to why this is not a major practical concern see Section~\ref{sec:experiments_setup}.

\begin{figure}
	\centering
	\begin{subfigure}{0.45\textwidth}
	\scalebox{0.4}{
\begin{tikzpicture}

\begin{axis}[
log basis y={10},
tick align=outside,
tick pos=left,
x grid style={white!69.0196078431373!black},
xmajorgrids,
xmin=-2.54, xmax=44.54,
xtick style={color=black},
y grid style={white!69.0196078431373!black},
ymajorgrids,
ymin=1.16433835690349, ymax=171770.511657778,
ymode=log,
ytick style={color=black},
ylabel={number of hidden symmetries},
xlabel={cnf instance}
]
\draw[line width=4pt, green!30] (axis cs:-10,100000) -- (axis cs:205,100000);
\draw[draw=none,fill=black] (axis cs:-0.4,0.01) rectangle (axis cs:0.4,3.99999796584257);
\draw[draw=none,fill=black] (axis cs:0.6,0.01) rectangle (axis cs:1.4,2);
\draw[draw=none,fill=black] (axis cs:1.6,0.01) rectangle (axis cs:2.4,4);
\draw[draw=none,fill=black] (axis cs:2.6,0.01) rectangle (axis cs:3.4,4);
\draw[draw=none,fill=black] (axis cs:3.6,0.01) rectangle (axis cs:4.4,2048.00540119679);
\draw[draw=none,fill=black] (axis cs:4.6,0.01) rectangle (axis cs:5.4,2048.00190408777);
\draw[draw=none,fill=black] (axis cs:5.6,0.01) rectangle (axis cs:6.4,6.00003208830702);
\draw[draw=none,fill=black] (axis cs:6.6,0.01) rectangle (axis cs:7.4,1.99998995308089);
\draw[draw=none,fill=black] (axis cs:7.6,0.01) rectangle (axis cs:8.4,1.99999938240027);
\draw[draw=none,fill=black] (axis cs:8.6,0.01) rectangle (axis cs:9.4,100000);
\draw[draw=none,fill=black] (axis cs:9.6,0.01) rectangle (axis cs:10.4,3.99999458111775);
\draw[draw=none,fill=black] (axis cs:10.6,0.01) rectangle (axis cs:11.4,2);
\draw[draw=none,fill=black] (axis cs:11.6,0.01) rectangle (axis cs:12.4,2);
\draw[draw=none,fill=black] (axis cs:12.6,0.01) rectangle (axis cs:13.4,2048.00034906293);
\draw[draw=none,fill=black] (axis cs:13.6,0.01) rectangle (axis cs:14.4,5.99999554329466);
\draw[draw=none,fill=black] (axis cs:14.6,0.01) rectangle (axis cs:15.4,2048.0059458758);
\draw[draw=none,fill=black] (axis cs:15.6,0.01) rectangle (axis cs:16.4,2);
\draw[draw=none,fill=black] (axis cs:16.6,0.01) rectangle (axis cs:17.4,2048.0095254417);
\draw[draw=none,fill=black] (axis cs:17.6,0.01) rectangle (axis cs:18.4,100000);
\draw[draw=none,fill=black] (axis cs:18.6,0.01) rectangle (axis cs:19.4,100000);
\draw[draw=none,fill=black] (axis cs:19.6,0.01) rectangle (axis cs:20.4,1.99999886971997);
\draw[draw=none,fill=black] (axis cs:20.6,0.01) rectangle (axis cs:21.4,2);
\draw[draw=none,fill=black] (axis cs:21.6,0.01) rectangle (axis cs:22.4,2048.00747320285);
\draw[draw=none,fill=black] (axis cs:22.6,0.01) rectangle (axis cs:23.4,2048.00547559622);
\draw[draw=none,fill=black] (axis cs:23.6,0.01) rectangle (axis cs:24.4,100000);
\draw[draw=none,fill=black] (axis cs:24.6,0.01) rectangle (axis cs:25.4,1.99999980569204);
\draw[draw=none,fill=black] (axis cs:25.6,0.01) rectangle (axis cs:26.4,100000);
\draw[draw=none,fill=black] (axis cs:26.6,0.01) rectangle (axis cs:27.4,2);
\draw[draw=none,fill=black] (axis cs:27.6,0.01) rectangle (axis cs:28.4,2048.00708490595);
\draw[draw=none,fill=black] (axis cs:28.6,0.01) rectangle (axis cs:29.4,2048.00438152381);
\draw[draw=none,fill=black] (axis cs:29.6,0.01) rectangle (axis cs:30.4,2048.01116931578);
\draw[draw=none,fill=black] (axis cs:30.6,0.01) rectangle (axis cs:31.4,100000);
\draw[draw=none,fill=black] (axis cs:31.6,0.01) rectangle (axis cs:32.4,2048.00332546336);
\draw[draw=none,fill=black] (axis cs:32.6,0.01) rectangle (axis cs:33.4,2048.00682204018);
\draw[draw=none,fill=black] (axis cs:33.6,0.01) rectangle (axis cs:34.4,3.99999429607856);
\draw[draw=none,fill=black] (axis cs:34.6,0.01) rectangle (axis cs:35.4,2.00002411294503);
\draw[draw=none,fill=black] (axis cs:35.6,0.01) rectangle (axis cs:36.4,2.00001205639984);
\draw[draw=none,fill=black] (axis cs:36.6,0.01) rectangle (axis cs:37.4,5.99999554329466);
\draw[draw=none,fill=black] (axis cs:37.6,0.01) rectangle (axis cs:38.4,4);
\draw[draw=none,fill=black] (axis cs:38.6,0.01) rectangle (axis cs:39.4,2);
\draw[draw=none,fill=black] (axis cs:39.6,0.01) rectangle (axis cs:40.4,2);
\draw[draw=none,fill=black] (axis cs:40.6,0.01) rectangle (axis cs:41.4,100000);
\draw[draw=none,fill=black] (axis cs:41.6,0.01) rectangle (axis cs:42.4,65536.0143737776);
\end{axis}

\end{tikzpicture}}
	\caption{Hidden symmetry.} \label{fig:hidden}
	\end{subfigure}
	\begin{subfigure}{0.45\textwidth}
		\scalebox{0.4}{
\begin{tikzpicture}

\begin{axis}[
log basis y={10},
tick align=outside,
tick pos=left,
x grid style={white!69.0196078431373!black},
xmajorgrids,
xmin=-4.34, xmax=82.34,
xtick style={color=black},
y grid style={white!69.0196078431373!black},
ymajorgrids,
ymin=1.16434449839209, ymax=171770.468513564,
ymode=log,
ytick style={color=black},
ylabel={number of reducible symmetries},
xlabel={cnf instance}
]
\draw[line width=4pt, green!30] (axis cs:-10,100000) -- (axis cs:205,100000);
\draw[draw=none,fill=white!50.1960784313725!black] (axis cs:-0.4,0.01) rectangle (axis cs:0.4,100000);
\draw[draw=none,fill=white!50.1960784313725!black] (axis cs:0.6,0.01) rectangle (axis cs:1.4,100000);
\draw[draw=none,fill=white!50.1960784313725!black] (axis cs:1.6,0.01) rectangle (axis cs:2.4,16);
\draw[draw=none,fill=white!50.1960784313725!black] (axis cs:2.6,0.01) rectangle (axis cs:3.4,100000);
\draw[draw=none,fill=white!50.1960784313725!black] (axis cs:3.6,0.01) rectangle (axis cs:4.4,100000);
\draw[draw=none,fill=white!50.1960784313725!black] (axis cs:4.6,0.01) rectangle (axis cs:5.4,2);
\draw[draw=none,fill=white!50.1960784313725!black] (axis cs:5.6,0.01) rectangle (axis cs:6.4,100000);
\draw[draw=none,fill=white!50.1960784313725!black] (axis cs:6.6,0.01) rectangle (axis cs:7.4,100000);
\draw[draw=none,fill=white!50.1960784313725!black] (axis cs:7.6,0.01) rectangle (axis cs:8.4,100000);
\draw[draw=none,fill=white!50.1960784313725!black] (axis cs:8.6,0.01) rectangle (axis cs:9.4,64);
\draw[draw=none,fill=white!50.1960784313725!black] (axis cs:9.6,0.01) rectangle (axis cs:10.4,100000);
\draw[draw=none,fill=white!50.1960784313725!black] (axis cs:10.6,0.01) rectangle (axis cs:11.4,8);
\draw[draw=none,fill=white!50.1960784313725!black] (axis cs:11.6,0.01) rectangle (axis cs:12.4,100000);
\draw[draw=none,fill=white!50.1960784313725!black] (axis cs:12.6,0.01) rectangle (axis cs:13.4,6);
\draw[draw=none,fill=white!50.1960784313725!black] (axis cs:13.6,0.01) rectangle (axis cs:14.4,767);
\draw[draw=none,fill=white!50.1960784313725!black] (axis cs:14.6,0.01) rectangle (axis cs:15.4,16);
\draw[draw=none,fill=white!50.1960784313725!black] (axis cs:15.6,0.01) rectangle (axis cs:16.4,64);
\draw[draw=none,fill=white!50.1960784313725!black] (axis cs:16.6,0.01) rectangle (axis cs:17.4,100000);
\draw[draw=none,fill=white!50.1960784313725!black] (axis cs:17.6,0.01) rectangle (axis cs:18.4,100000);
\draw[draw=none,fill=white!50.1960784313725!black] (axis cs:18.6,0.01) rectangle (axis cs:19.4,100000);
\draw[draw=none,fill=white!50.1960784313725!black] (axis cs:19.6,0.01) rectangle (axis cs:20.4,100000);
\draw[draw=none,fill=white!50.1960784313725!black] (axis cs:20.6,0.01) rectangle (axis cs:21.4,100000);
\draw[draw=none,fill=white!50.1960784313725!black] (axis cs:21.6,0.01) rectangle (axis cs:22.4,2);
\draw[draw=none,fill=white!50.1960784313725!black] (axis cs:22.6,0.01) rectangle (axis cs:23.4,2);
\draw[draw=none,fill=white!50.1960784313725!black] (axis cs:23.6,0.01) rectangle (axis cs:24.4,100000);
\draw[draw=none,fill=white!50.1960784313725!black] (axis cs:24.6,0.01) rectangle (axis cs:25.4,100000);
\draw[draw=none,fill=white!50.1960784313725!black] (axis cs:25.6,0.01) rectangle (axis cs:26.4,100000);
\draw[draw=none,fill=white!50.1960784313725!black] (axis cs:26.6,0.01) rectangle (axis cs:27.4,100000);
\draw[draw=none,fill=white!50.1960784313725!black] (axis cs:27.6,0.01) rectangle (axis cs:28.4,100000);
\draw[draw=none,fill=white!50.1960784313725!black] (axis cs:28.6,0.01) rectangle (axis cs:29.4,100000);
\draw[draw=none,fill=white!50.1960784313725!black] (axis cs:29.6,0.01) rectangle (axis cs:30.4,100000);
\draw[draw=none,fill=white!50.1960784313725!black] (axis cs:30.6,0.01) rectangle (axis cs:31.4,2);
\draw[draw=none,fill=white!50.1960784313725!black] (axis cs:31.6,0.01) rectangle (axis cs:32.4,100000);
\draw[draw=none,fill=white!50.1960784313725!black] (axis cs:32.6,0.01) rectangle (axis cs:33.4,2);
\draw[draw=none,fill=white!50.1960784313725!black] (axis cs:33.6,0.01) rectangle (axis cs:34.4,4);
\draw[draw=none,fill=white!50.1960784313725!black] (axis cs:34.6,0.01) rectangle (axis cs:35.4,384);
\draw[draw=none,fill=white!50.1960784313725!black] (axis cs:35.6,0.01) rectangle (axis cs:36.4,2);
\draw[draw=none,fill=white!50.1960784313725!black] (axis cs:36.6,0.01) rectangle (axis cs:37.4,100000);
\draw[draw=none,fill=white!50.1960784313725!black] (axis cs:37.6,0.01) rectangle (axis cs:38.4,256);
\draw[draw=none,fill=white!50.1960784313725!black] (axis cs:38.6,0.01) rectangle (axis cs:39.4,767);
\draw[draw=none,fill=white!50.1960784313725!black] (axis cs:39.6,0.01) rectangle (axis cs:40.4,16);
\draw[draw=none,fill=white!50.1960784313725!black] (axis cs:40.6,0.01) rectangle (axis cs:41.4,100000);
\draw[draw=none,fill=white!50.1960784313725!black] (axis cs:41.6,0.01) rectangle (axis cs:42.4,100000);
\draw[draw=none,fill=white!50.1960784313725!black] (axis cs:42.6,0.01) rectangle (axis cs:43.4,100000);
\draw[draw=none,fill=white!50.1960784313725!black] (axis cs:43.6,0.01) rectangle (axis cs:44.4,100000);
\draw[draw=none,fill=white!50.1960784313725!black] (axis cs:44.6,0.01) rectangle (axis cs:45.4,100000);
\draw[draw=none,fill=white!50.1960784313725!black] (axis cs:45.6,0.01) rectangle (axis cs:46.4,100000);
\draw[draw=none,fill=white!50.1960784313725!black] (axis cs:46.6,0.01) rectangle (axis cs:47.4,16);
\draw[draw=none,fill=white!50.1960784313725!black] (axis cs:47.6,0.01) rectangle (axis cs:48.4,100000);
\draw[draw=none,fill=white!50.1960784313725!black] (axis cs:48.6,0.01) rectangle (axis cs:49.4,100000);
\draw[draw=none,fill=white!50.1960784313725!black] (axis cs:49.6,0.01) rectangle (axis cs:50.4,100000);
\draw[draw=none,fill=white!50.1960784313725!black] (axis cs:50.6,0.01) rectangle (axis cs:51.4,2);
\draw[draw=none,fill=white!50.1960784313725!black] (axis cs:51.6,0.01) rectangle (axis cs:52.4,2);
\draw[draw=none,fill=white!50.1960784313725!black] (axis cs:52.6,0.01) rectangle (axis cs:53.4,100000);
\draw[draw=none,fill=white!50.1960784313725!black] (axis cs:53.6,0.01) rectangle (axis cs:54.4,100000);
\draw[draw=none,fill=white!50.1960784313725!black] (axis cs:54.6,0.01) rectangle (axis cs:55.4,23);
\draw[draw=none,fill=white!50.1960784313725!black] (axis cs:55.6,0.01) rectangle (axis cs:56.4,100000);
\draw[draw=none,fill=white!50.1960784313725!black] (axis cs:56.6,0.01) rectangle (axis cs:57.4,100000);
\draw[draw=none,fill=white!50.1960784313725!black] (axis cs:57.6,0.01) rectangle (axis cs:58.4,100000);
\draw[draw=none,fill=white!50.1960784313725!black] (axis cs:58.6,0.01) rectangle (axis cs:59.4,2);
\draw[draw=none,fill=white!50.1960784313725!black] (axis cs:59.6,0.01) rectangle (axis cs:60.4,100000);
\draw[draw=none,fill=white!50.1960784313725!black] (axis cs:60.6,0.01) rectangle (axis cs:61.4,6);
\draw[draw=none,fill=white!50.1960784313725!black] (axis cs:61.6,0.01) rectangle (axis cs:62.4,100000);
\draw[draw=none,fill=white!50.1960784313725!black] (axis cs:62.6,0.01) rectangle (axis cs:63.4,100000);
\draw[draw=none,fill=white!50.1960784313725!black] (axis cs:63.6,0.01) rectangle (axis cs:64.4,100000);
\draw[draw=none,fill=white!50.1960784313725!black] (axis cs:64.6,0.01) rectangle (axis cs:65.4,5);
\draw[draw=none,fill=white!50.1960784313725!black] (axis cs:65.6,0.01) rectangle (axis cs:66.4,100000);
\draw[draw=none,fill=white!50.1960784313725!black] (axis cs:66.6,0.01) rectangle (axis cs:67.4,5);
\draw[draw=none,fill=white!50.1960784313725!black] (axis cs:67.6,0.01) rectangle (axis cs:68.4,40320);
\draw[draw=none,fill=white!50.1960784313725!black] (axis cs:68.6,0.01) rectangle (axis cs:69.4,3);
\draw[draw=none,fill=white!50.1960784313725!black] (axis cs:69.6,0.01) rectangle (axis cs:70.4,100000);
\draw[draw=none,fill=white!50.1960784313725!black] (axis cs:70.6,0.01) rectangle (axis cs:71.4,100000);
\draw[draw=none,fill=white!50.1960784313725!black] (axis cs:71.6,0.01) rectangle (axis cs:72.4,100000);
\draw[draw=none,fill=white!50.1960784313725!black] (axis cs:72.6,0.01) rectangle (axis cs:73.4,2520);
\draw[draw=none,fill=white!50.1960784313725!black] (axis cs:73.6,0.01) rectangle (axis cs:74.4,100000);
\draw[draw=none,fill=white!50.1960784313725!black] (axis cs:74.6,0.01) rectangle (axis cs:75.4,100000);
\draw[draw=none,fill=white!50.1960784313725!black] (axis cs:75.6,0.01) rectangle (axis cs:76.4,6);
\draw[draw=none,fill=white!50.1960784313725!black] (axis cs:76.6,0.01) rectangle (axis cs:77.4,100000);
\draw[draw=none,fill=white!50.1960784313725!black] (axis cs:77.6,0.01) rectangle (axis cs:78.4,2);
\end{axis}

\end{tikzpicture}}
	\caption{Reducible symmetry.} \label{fig:reducible}
	\end{subfigure}
	\caption{The diagrams show instances where non-trivial symmetries of the respective type exist (a value of $1$ would mean no non-trivial symmetries existed). The left diagram shows the number of hidden symmetries detected in the instance. The right diagram shows instances where the initial group size was larger than the preprocessed group size, i.e., the number of reducible symmetries in the instance. Both diagrams are cut off at $10^5$ for clarity (values go beyond $10^{2000}$ in both diagrams).} \label{fig:sym_red_hid}
\end{figure}
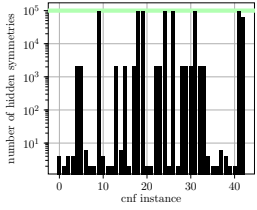
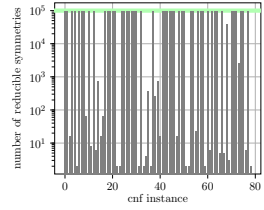

\textbf{Detected Symmetry.} There are substantial differences in the detected symmetry of the unprocessed and preprocessed instances.

Using our preprocessor, we uncovered hidden symmetry in $43$ instances (see Figure~\ref{fig:hidden}). 
We found reducible symmetry in $73$ instances.
Indeed, these were often very large groups (see Figure~\ref{fig:reducible}) that in many cases even included all symmetries of the respective instance. 

\section{Conclusion and Future Work}
Unit, pure, subsumption and blocked clause elimination are symmetry-preserving.
If these simplifications are to be applied to a formula anyway, then symmetries should be detected and exploited after simplifying the formula.
Other techniques, such us adding clauses to the formula derived using resolution, self-subsumption and variable elimination, turn out to not be symmetry-preserving, and can potentially remove symmetry from formulas in an undesirable manner.
Going beyond the analysis, for variable elimination and self-subsumption we defined restricted variants, which enable the rules to be applied in a symmetry-preserving manner.

In practice, instances simplified using a symmetry-preserving preprocessor are substantially easier to handle for symmetry detection tools. 
In fact, the structure of instances changes considerably and techniques previously designed for CNF formulas seem to become less impactful.
Most importantly, this opens up the opportunity to tune symmetry detection tools to solve \emph{preprocessed} CNF formulas (e.g., to address the shortcomings we raised with the graph encoding in Section~\ref{sec:impl}).

Regarding symmetries, preprocessed instances, due to the symmetry-preserving nature of transformations, are guaranteed to contain at least all applicable symmetries of the unprocessed instances. But indeed, in $21$\% of the symmetrical instances tested they even contained \emph{more} symmetry.
We also found reducible symmetry, that exclusively interacts with literals removable through preprocessing, in $40$\% of the benchmark instances.
Overall, we believe that this motivates an even deeper analysis into how much ``exploitable'' symmetry is in the instances, and how it can be systematically uncovered.
This could also involve tuning and testing with state-of-the-art preprocessing, SAT solvers and symmetry exploitation in the loop: not only to gain a better understanding of the potential effects of the different types of symmetry, but also the interaction between algorithms.

There are even more avenues to expand upon or apply the work in this paper. 
For example, there are many other preprocessing techniques in the literature (e.g., bounded variable addition) and dynamic techniques (e.g., inprocessing using learned clauses) that could be analyzed.

\section*{Acknowledgements}
I want to thank Moritz Lichter, Pascal Schweitzer, Constantin Seebach and Damien Zufferey 
for the valuable discussions at different stages of the project.
I also want to thank the anonymous reviewers at SAT 2022 for pointing out an error with the counter-examples used in an earlier version of the paper.

\bibliography{main}
\bibliographystyle{plain}
\end{document}